\documentclass[a4paper,11pt]{article}
\pdfoutput=1

\usepackage{amsfonts,amsmath,amssymb,amsthm,booktabs,color,doi,enumitem,graphicx,latexsym,marginnote,microtype,multirow,wrapfig}
\usepackage[a4paper,vmargin=28mm,hmargin=38mm]{geometry}
\usepackage[numbers,sort&compress]{natbib}
\usepackage[basic]{complexity}
\usepackage[font=small,margin=1em]{caption}
\usepackage{hyperref}

\newcommand{\Lo}{\mathcal{L}} 
\newcommand{\aA}{\mathcal{A}} 
\newcommand{\aB}{\mathcal{B}} 
\newcommand{\aC}{\mathcal{C}} 
\newcommand{\sA}{\mathbf{A}} 
\newcommand{\sB}{\mathbf{B}} 
\newcommand{\F}{\mathcal{F}} 
\newcommand{\Pp}{\mathcal{P}} 
\newcommand{\N}{\mathbb{N}}
\DeclareMathOperator{\set}{set}
\DeclareMathOperator{\multiset}{multiset}

\newclass{\Vector}{Vector}
\newclass{\Multiset}{Multiset}
\newclass{\Set}{Set}
\newclass{\Broadcast}{Broadcast}

\newcommand{\sVector}{\mathbf{Vector}}
\newcommand{\sMultiset}{\mathbf{Multiset}}
\newcommand{\sSet}{\mathbf{Set}}
\newcommand{\sBroadcast}{\mathbf{Broadcast}}

\renewcommand{\models}{\vDash}
\newcommand{\nmodels}{\nvDash}

\newcommand{\VVc}{{\ComplexityFont{VV}_{\ComplexityFont{c}}}}
\newclass{\VV}{VV}
\newclass{\MV}{MV}
\newclass{\SV}{SV}
\newclass{\VB}{VB}
\newclass{\MB}{MB}
\newclass{\SB}{SB}
\newcommand{\SBo}{{\ComplexityFont{SB}_{\ComplexityFont{o}}}}
\newclass{\LOCAL}{LOCAL}

\newcommand{\VVcl}{\VVc(1)}
\newcommand{\VVl}{\VV(1)}
\newcommand{\MVl}{\MV(1)}
\newcommand{\SVl}{\SV(1)}
\newcommand{\VBl}{\VB(1)}
\newcommand{\MBl}{\MB(1)}
\newcommand{\SBl}{\SB(1)}

\newcommand{\ML}{\mathrm{ML}} 
\newcommand{\GML}{\mathrm{GML}} 
\newcommand{\MML}{\mathrm{MML}} 
\newcommand{\GMML}{\mathrm{GMML}} 
\newcommand{\KM}{\mathit{K}} 
\newcommand{\CM}{\mathcal{K}} 
\newcommand{\md}{\mathrm{md}} 
\newcommand{\pow}{\mathcal{P}} 

\newcommand{\KVV}{\mathit{K}_{+,+}} 
\newcommand{\KMV}{\mathit{K}_{-,+}}
\newcommand{\KVB}{\mathit{K}_{+,-}}
\newcommand{\KMB}{\mathit{K}_{-,-}}
\newcommand{\CVVc}{\mathcal{K}^\mathsf{c}_{+,+}} 
\newcommand{\CVV}{\mathcal{K}_{+,+}} 
\newcommand{\CMV}{\mathcal{K}_{-,+}} 
\newcommand{\CVB}{\mathcal{K}_{+,-}} 
\newcommand{\CMB}{\mathcal{K}_{-,-}} 

\newtheorem{theorem}{Theorem}
\newtheorem{lemma}[theorem]{Lemma}
\newtheorem{fact}[theorem]{Fact}
\newtheorem{corollary}[theorem]{Corollary}

\theoremstyle{remark}
\newtheorem{remark}{Remark}

\setenumerate{label={\normalfont(\alph*)},leftmargin=7ex}
\setitemize{label=--}
\hypersetup{
    colorlinks=true,
    linkcolor=black,
    citecolor=black,
    filecolor=black,
    urlcolor=[rgb]{0,0.1,0.5},
}

\newcounter{myexternalpagenum}
\newcommand{\definepage}[1]{\stepcounter{myexternalpagenum}\edef#1{\arabic{myexternalpagenum}}}
\definepage{\PPortNumbering}
\definepage{\PConsistent}
\definepage{\PInput}
\definepage{\POutput}
\definepage{\PBeforeAfter}
\definepage{\PSummary}
\definepage{\PRelations}
\definepage{\PVvVvc}
\definepage{\PSumDeg}
\definepage{\PRegularVv}

\hypersetup{
    pdftitle={Weak Models of Distributed Computing, with Connections to Modal Logic},
    pdfauthor={Lauri Hella, Matti J\"arvisalo, Antti Kuusisto, Juhana Laurinharju, Tuomo Lempi\"ainen, Kerkko Luosto, Jukka Suomela, Jonni Virtema},
}

\begin{document}

\vspace*{5ex}
\begin{center}
    {\Large \textbf{Weak Models of Distributed Computing, \\ with Connections to Modal Logic}\par}

    \bigskip
    \bigskip

    Lauri Hella$^1$,
    Matti J\"arvisalo$^2$,
    Antti Kuusisto$^3$,
    Juhana Laurinharju$^2$, \\
    Tuomo Lempi\"ainen$^2$,
    Kerkko Luosto$^1$,
    Jukka Suomela$^2$, and
    Jonni Virtema$^1$

    \bigskip
    $^1$School of Information Sciences, University of Tampere \\[1ex]
    $^2$Helsinki Institute for Information Technology HIIT, \\ Department of Computer Science, University of Helsinki \\[1ex]
    $^3$Institute of Computer Science, University of Wroc\l aw, Poland
\end{center}

\bigskip
\bigskip
\noindent\textbf{Abstract.}
This work presents a classification of weak models of distributed computing. We focus on deterministic distributed algorithms, and study models of computing that are weaker versions of the widely-studied port-numbering model. In the port-numbering model, a node of degree $d$ receives messages through $d$ input ports and sends messages through $d$ output ports, both numbered with $1,2,\dotsc,d$. In this work, $\VVc$ is the class of all \emph{graph problems} that can be solved in the standard port-numbering model. We study the following subclasses of $\VVc$:
\begin{itemize}[noitemsep,leftmargin=4em]
    \item[$\VV$:] Input port $i$ and output port $i$ are not necessarily connected to the same neighbour.
    \item[$\MV$:] Input ports are not numbered; algorithms receive a multiset of messages.
    \item[$\SV$:] Input ports are not numbered; algorithms receive a set of messages.
    \item[$\VB$:] Output ports are not numbered; algorithms send the same message to all output ports.
    \item[$\MB$:] Combination of $\MV$ and $\VB$.
    \item[$\SB$:] Combination of $\SV$ and $\VB$.
\end{itemize}
Now we have many trivial containment relations, such as $\SB \subseteq \MB \subseteq \VB \subseteq \VV \subseteq \VVc$, but it is not obvious if, for example, either of $\VB \subseteq \SV$ or $\SV \subseteq \VB$ should hold. Nevertheless, it turns out that we can identify a \emph{linear order} on these classes. We prove that $\SB \subsetneq \MB = \VB \subsetneq \SV = \MV = \VV \subsetneq \VVc$. The same holds for the constant-time versions of these classes.

We also show that the constant-time variants of these classes can be characterised by a corresponding \emph{modal logic}. Hence the linear order identified in this work has direct implications in the study of the expressibility of modal logic. Conversely, one can use tools from modal logic to study these classes.

\thispagestyle{empty}
\setcounter{page}{0}
\newpage

\section{Introduction}

We introduce seven complexity classes, $\VVc$, $\VV$, $\MV$, $\SV$, $\VB$, $\MB$, and $\SB$, each defined as the class of \emph{graph problems} that can be solved with a deterministic distributed algorithm in a certain variant of the widely-studied \emph{port-numbering model}. We present a \emph{complete characterisation} of the containment relations between these classes, as well as their constant-time counterparts, and identify connections between these classes and questions related to \emph{modal logic}.

\subsection{State Machines}\label{ssec:distalg}

For our purposes, a distributed algorithm is best understood as a state machine~$\aA$. In a distributed system, each node is a copy of the same state machine~$\aA$. Computation proceeds in synchronous steps. In each step, each machine
\begin{enumerate}[label=(\arabic*),noitemsep]
    \item sends messages to its neighbours,
    \item receives messages from its neighbours, and
    \item updates its state based on the messages that it received. 
\end{enumerate}
If the new state is a stopping state, the machine halts.

Let us now formalise the setting studied in this work. We use the notation $[k] = \{1,2,\dotsc,k\}$. For each positive integer $\Delta$, let $\F(\Delta)$ consist of all simple undirected graphs of maximum degree at most~$\Delta$. A \emph{distributed state machine} for $\F(\Delta)$ is a tuple $\aA = (Y, Z, z_0, M, m_0, \mu, \delta)$, where
\begin{itemize}[noitemsep]
    \item $Y$ is a finite set of stopping states,
    \item $Z$ is a (possibly infinite) set of intermediate states such that $Y\cap Z=\emptyset$,
    \item $z_0\colon \{0,1,\dotsc,\Delta\} \to Y\cup Z$ defines the initial state depending on the degree of the node,
    \item $M$ is a (possibly infinite) set of messages,
    \item $m_0 \in M$ is a special symbol for ``no message'',
    \item $\mu\colon Z \times [\Delta] \to M$ is a function that constructs the outgoing messages,
    \item $\delta\colon Z \times M^\Delta \to Y \cup Z$ defines the state transitions.
\end{itemize}
To simplify the notation, we extend the domains of $\mu$ and $\delta$ to cover the stopping states: for all $y \in Y$, we define $\mu(y,i) = m_0$ for any $i \in [\Delta]$, and $\delta(y,\vec{m}) = y$ for any $\vec{m} \in M^\Delta$. In other words, a node that has stopped does not send any messages and does not change its state any more.

\subsection{Port Numbering}

Now consider a graph $G = (V,E) \in \F(\Delta)$. We write $\deg(v)$ for the degree of node $v \in V$. A \emph{port} of $G$ is a pair $(v,i)$ where $v \in V$ and $i \in [\deg(v)]$. Let $P(G)$ be the set of all ports of $G$. Let $p\colon P(G) \to P(G)$ be a bijection. Define
\begin{align*}
    A(p) &= \{ (u,v) : \text{$u \in V$, $v \in V$, and $p((u,i)) = (v,j)$ for some $i$ and $j$} \}, \\
    A(G) &= \{ (u,v) : \{u,v\} \in E \}.
\end{align*}
We say that $p$ is a \emph{port numbering} of $G$ if $A(p) = A(G)$; see Figure~\ref{fig:port-numbering} for an example. The intuition here is that a node $v \in V$ has $\deg(v)$ communication ports; if it sends a message to its port $(v,i)$, and $p((v,i)) = (u,j)$, the message will be received by its neighbour $u$ from port $(u,j)$.

\begin{figure}[t]
    \centering
    \includegraphics[page=\PPortNumbering]{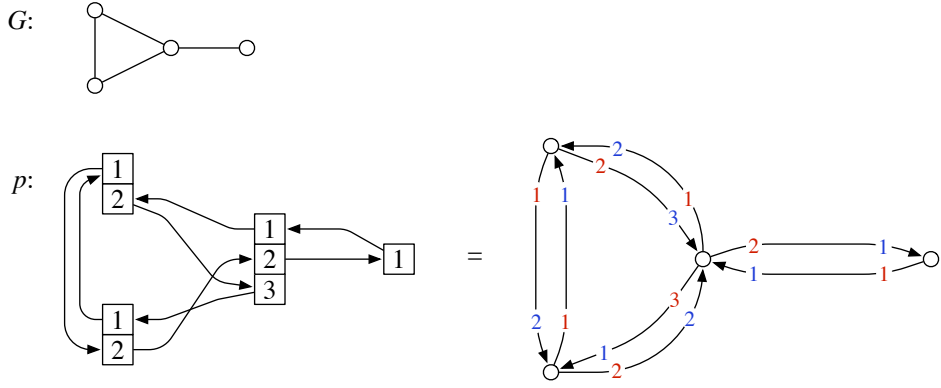}
    \caption{A port numbering $p$ of graph $G$. Here we present $p$ using two different notations; in the illustration on the left, the ports are explicitly shown, while in the illustration on the right, the ports are given as the labels of the edges.}\label{fig:port-numbering}
\end{figure}

We say that a port numbering is \emph{consistent} if $p$ is an involution, that is,
\[
    p\bigl(p((v,i))\bigr) = (v,i) \text{ for all } (v,i) \in P(G).
\]
See Figure~\ref{fig:consistent} for an example.

\begin{figure}[t]
    \centering
    \includegraphics[page=\PConsistent]{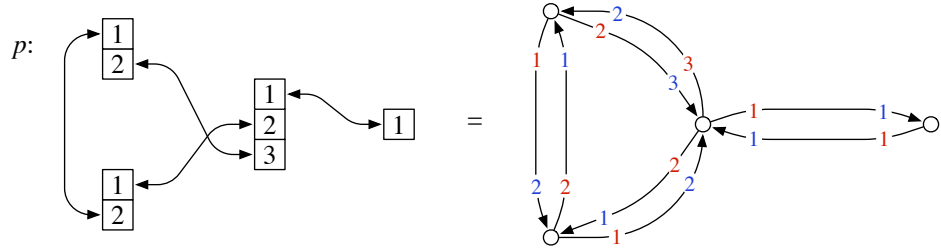}
    \caption{A consistent port numbering.}\label{fig:consistent}
\end{figure}

\subsection{Execution of a State Machine}\label{execution}

For a fixed distributed state machine $\aA$, a graph $G$, and a port numbering $p$, we can define the \emph{execution} of $\aA$ in $(G,p)$ recursively as follows.

The state of the system at time $t = 0, 1, \dotsc$ is represented as a state vector $x_t\colon V \to Y \cup Z$. At time $0$, we have
\[
    x_0(u) = z_0(\deg(u))
\]
for each $u \in V$.

Now assume that we have defined the state $x_t$ at time $t$. Let $(u,i) \in P(G)$ and $(v,j) = p^{-1}((u,i))$. Define
\[
    a_{t+1}(u,i) = \mu(x_t(v),j).
\]
In words, $a_{t+1}(u,i)$ is the message received by node $u$ from port $(u,i)$ in round $t+1$, or equivalently the message sent by node $v$ to port $(v,j)$. For each $u \in V$ we define a vector $\vec{a}_{t+1}(u)$ of length $\Delta$ as follows:
\[
    \vec{a}_{t+1}(u) = \bigl( a_{t+1}(u,1),\, a_{t+1}(u,2),\, \dotsc,\, a_{t+1}(u,\deg(u)),\, m_0, m_0, \dotsc, m_0 \bigr).
\]
In other words, we simply take all messages received by $u$, in the order of increasing port number; the padding with the dummy messages $m_0$ is just for notational convenience so that $\vec{a}_{t+1}(u) \in M^\Delta$. Finally, we define the new state of a node $u \in V$ as follows:
\[
    x_{t+1}(u) = \delta(x_t(u), \vec{a}_{t+1}(u)).
\]

We say that $\aA$ \emph{stops in time $T$} in $(G,p)$ if $x_T(u) \in Y$ for all $u \in V$. If $\aA$ stops in time $T$ in $(G,p)$, we say that $S = x_T$ is the \emph{output} of $\aA$ in $(G,p)$. Here $S(u) = x_T(u)$ is the \emph{local output} of $u \in V$.

\subsection{Graph Problems}

A \emph{graph problem} is a function $\Pi$ that associates with each undirected graph $G = (V,E)$ a set $\Pi(G)$ of \emph{solutions}. Each solution $S \in \Pi(G)$ is a mapping $S\colon V \to Y$; here $Y$ is a finite set that does not depend on $G$.

We emphasise that this definition is by no means universal; however, it is convenient for our purposes and covers a wide range of classical graph problems:
\begin{itemize}
    \item Finding a \emph{subset of vertices}. A typical example is the task of finding a \emph{maximal independent set}: $Y = \{0,1\}$, and each solution $S$ is the indicator function of a maximal independent set.
    \item Finding a \emph{partition of vertices}. A typical example is the task of finding a \emph{vertex $3$-colouring}: $Y = \{1,2,3\}$, and each solution $S$ is a valid $3$-colouring of the graph.
    \item Deciding \emph{graph properties}. A typical example is deciding if a graph is \emph{Eulerian}: Here $Y = \{0,1\}$. If $G$ is Eulerian, there is only one solution $S$ with $S(v) = 1$ for all $v \in V$. If $G$ is not Eulerian, valid solutions are mappings $S$ such that $S(v) = 0$ for at least one $v \in V$. Put otherwise, all nodes must accept a yes-instance, and at least one node must reject a no-instance.
\end{itemize}
The idea is that a distributed state machine $\aA$ \emph{solves} a graph problem $\Pi$ if, for any graph $G$ and for any port numbering of $G$, the output of $\aA$ is a valid solution $S \in \Pi(G)$. However, the fact that we study graphs of bounded degree requires some care; hence the following somewhat technical definition.

Let $\Pi$ be a graph problem. Let $T\colon \N \times \N \to \N$. Let $\sA = (\aA_1, \aA_2, \dotsc)$ be a sequence of distributed state machines. We say that \emph{$\sA$ solves $\Pi$ in time $T$} if the following hold for any $\Delta \in \N$, any graph $G \in \F(\Delta)$, and any port numbering $p$ of $G$:
\begin{enumerate}[noitemsep]
    \item State machine $\aA_\Delta$ stops in time $T(\Delta,|V|)$ in $(G,p)$.
    \item The output of $\aA_\Delta$ is in $\Pi(G)$.
\end{enumerate}

We say that \emph{$\sA$ solves $\Pi$ in time $T$ assuming consistency} if the above holds for any consistent port numbering $p$ of $G$. Note that we do not require that $\aA_\Delta$ stops if the port numbering happens to be inconsistent.

We say that \emph{$\sA$ solves $\Pi$} or \emph{$\sA$ is an algorithm for $\Pi$} if there is any function $T$ such that $\sA$ solves~$\Pi$ in time $T$. We say that \emph{$\sA$ solves $\Pi$ in constant time} or \emph{$\sA$ is a local algorithm for $\Pi$} if $T(\Delta,n) = T'(\Delta)$ for some $T'\colon \N \to \N$, independently of $n$.

\begin{remark}
We emphasise that the term ``constant time'' refers to the case of a fixed $\Delta$. We only require that for each given $\Delta$ the running time of state machine $\aA_\Delta$ on graph family $\F(\Delta)$ is bounded by a constant. That is, ``local algorithms'' are $O(1)$-time algorithms on any graph family of maximum degree $O(1)$.
\end{remark}

\subsection{Algorithm Classes}

Now we are ready to introduce the concepts studied in this work: variants of the definition of a distributed algorithm.

For a vector $\vec{a} = (a_1, a_2, \dotsc, a_\Delta) \in M^\Delta$ we define
\begin{align*}
    \set(\vec{a}) &= \{ a_1, a_2, \dotsc, a_\Delta \}, \\
    \multiset(\vec{a}) &= {\bigr\{ (m,n) : m \in M,\ n = | \{ i \in [\Delta] : m = a_i \} | \bigl\}}.
\end{align*}
In other words, $\multiset(\vec{a})$ discards the ordering of the elements of $\vec{a}$, and $\set(\vec{a})$ furthermore discards the multiplicities.

Let $\Vector$ be the set of all distributed state machines $\aA$, as defined in Section~\ref{ssec:distalg}. We define three subclasses of distributed state machines, $\Set \subseteq \Multiset \subseteq \Vector$, and $\Broadcast \subseteq \Vector$:
\begin{itemize}
    \item $\aA \in \Multiset$ if $\multiset(\vec{a}) = \multiset(\vec{b})$ implies $\delta(x,\vec{a}) = \delta(x,\vec{b})$ for all $x \in Z$.
    \item $\aA \in \Set$ if $\set(\vec{a}) = \set(\vec{b})$ implies $\delta(x,\vec{a}) = \delta(x,\vec{b})$ for all $x \in Z$.
    \item $\aA \in \Broadcast$ if $\mu(x,i) = \mu(x,j)$ for all $x \in Z$ and $i, j \in [\Delta]$.
\end{itemize}

Classes $\Multiset$ and $\Set$ are related to \emph{incoming} messages; see Figure~\ref{fig:input} for an example. Intuitively, a state machine in class $\Vector$ considers a \emph{vector} of incoming messages, while a state machine in $\Multiset$ considers a \emph{multiset} of incoming messages, and a state machine in $\Set$ considers a \emph{set} of incoming messages. In particular, state machines in $\Multiset$ and $\Set$ do not have any access to the numbering of incoming ports.

\begin{figure}
    \centering
    \includegraphics[page=\PInput]{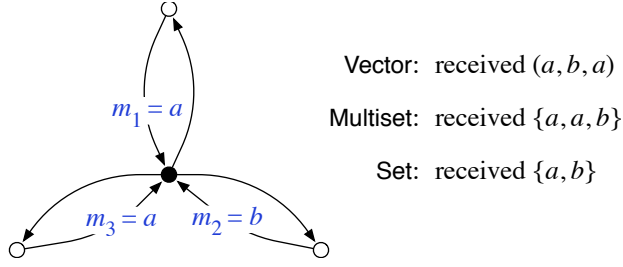}
    \caption{Comparison of $\Vector$, $\Multiset$, and $\Set$.}\label{fig:input}
\end{figure}

Class $\Broadcast$ is related to \emph{outgoing} messages; see Figure~\ref{fig:output} for an example. Intuitively, a state machine in class $\Vector$ constructs a \emph{vector} of outgoing messages, while a state machine in $\Broadcast$ can only \emph{broadcast} the same message to all neighbours. In particular, state machines in $\Broadcast$ do not have any access to the numbering of outgoing ports.

\begin{figure}
    \centering
    \includegraphics[page=\POutput]{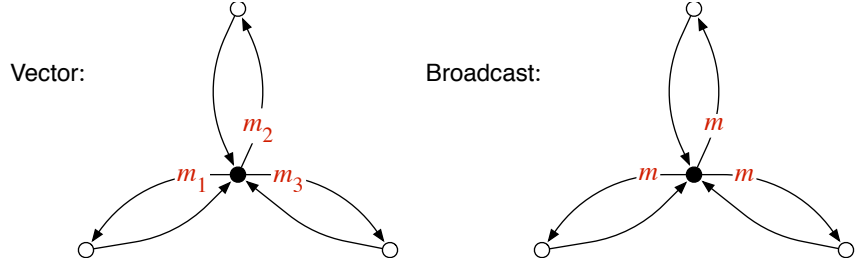}
    \caption{Comparison of $\Vector$ and $\Broadcast$.}\label{fig:output}
\end{figure}

We extend the definitions to sequences of state machines in a natural way:
\begin{align*}
    \sVector &= \bigl\{ (\aA_1, \aA_2, \dotsc) \,:\, \aA_\Delta \in \Vector \text{ for all $\Delta$} \bigr\}, \\
    \sMultiset &= \bigl\{ (\aA_1, \aA_2, \dotsc) \,:\, \aA_\Delta \in \Multiset \text{ for all $\Delta$} \bigr\}, \\
    \sSet &= \bigl\{ (\aA_1, \aA_2, \dotsc) \,:\, \aA_\Delta \in \Set \text{ for all $\Delta$} \bigr\}, \\
    \sBroadcast &= \bigl\{ (\aA_1, \aA_2, \dotsc) \,:\, \aA_\Delta \in \Broadcast \text{ for all $\Delta$} \bigr\}.
\end{align*}

From now on, we will use the word \emph{algorithm} to refer to both distributed state machines $\aA \in \Vector$ and to sequences of distributed state machines $\sA \in \sVector$, when there is no risk of confusion.

\subsection{Problem Classes}

So far we have defined classes of algorithms; now we will define seven classes of problems:
\begin{enumerate}
    \item $\Pi \in \VVc$ if there is an algorithm $\sA \in \sVector$ that solves problem~$\Pi$ assuming consistency,
    \item $\Pi \in \VV$ if there is an algorithm $\sA \in \sVector$ that solves problem~$\Pi$,
    \item $\Pi \in \MV$ if there is an algorithm $\sA \in \sMultiset$ that solves problem~$\Pi$,
    \item $\Pi \in \SV$ if there is an algorithm $\sA \in \sSet$ that solves problem~$\Pi$,
    \item $\Pi \in \VB$ if there is an algorithm $\sA \in \sBroadcast$ that solves problem~$\Pi$,
    \item $\Pi \in \MB$ if there is an algorithm $\sA \in \sMultiset \cap \sBroadcast$ that solves problem~$\Pi$,
    \item $\Pi \in \SB$ if there is an algorithm $\sA \in \sSet \cap \sBroadcast$ that solves problem~$\Pi$.
\end{enumerate}
We will also define the constant-time variants of the classes:
\begin{enumerate}
    \item $\Pi \in \VVcl$ if there is a local algorithm $\sA \in \sVector$ that solves problem~$\Pi$ assuming consistency,
    \item $\Pi \in \VVl$ if there is a local algorithm $\sA \in \sVector$ that solves problem~$\Pi$,
    \\[0.5ex]\ldots
\end{enumerate}
Note that consistency is irrelevant for all other classes; we only define the consistent variants of $\VV$ and $\VVl$. The classes are summarised in Figure~\ref{fig:beforeafter}a. Figure~\ref{fig:summary} summarises what information is available to an algorithm in each class.

\begin{remark}\label{rem:sbo}
In each problem class, we consider algorithms in which each node knows its own degree. While this is natural in all other cases, it may seem odd in the case of class $\SB$. In principle, we could define yet another class of problems $\SBo$, defined in terms of \emph{degree-oblivious} algorithms in $\sSet \cap \sBroadcast$, i.e., algorithms with a constant initialisation function $z_0$. However, it is easy to see that $\SBo$ is entirely trivial---in essence, one can only solve the problem of distinguishing non-isolated nodes from isolated nodes---while there are many non-trivial problems that we can solve in class $\SB$. In particular, it is trivial to prove that $\SBo \subsetneq \SB$. Hence we will not consider class $\SBo$ in this work. However, class $\SBo$ is more interesting if one considers labelled graphs; see Section~\ref{ssec:localinput}.
\end{remark}

\begin{figure}
    \centering
    \includegraphics[page=\PBeforeAfter]{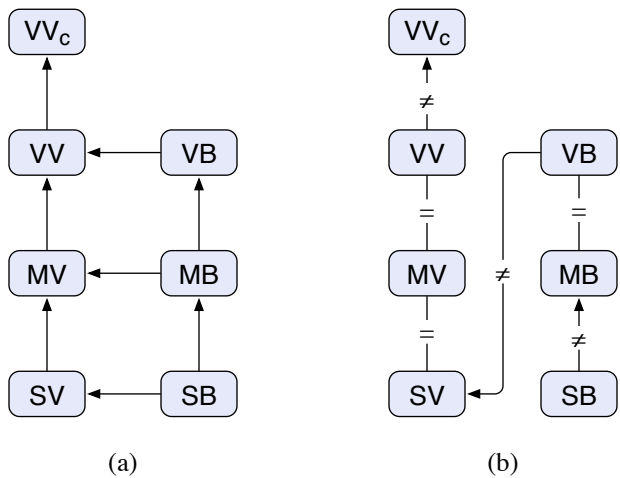}
    \caption{Classes of graph problems. (a)~Trivial subset relations between the classes. (b)~The linear order identified in this work.}\label{fig:beforeafter}
\end{figure}

\begin{figure}
    \centering
    \includegraphics[page=\PSummary]{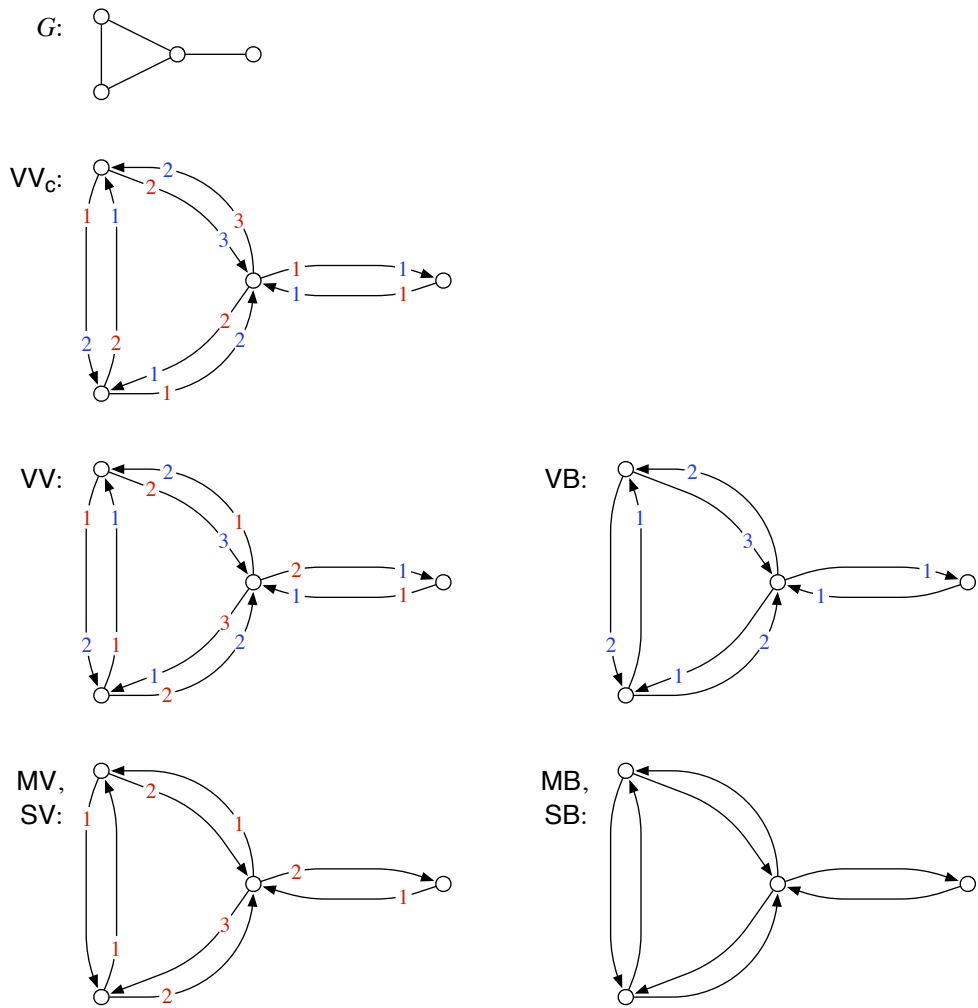}
    \caption{Auxiliary information available to a distributed algorithm in each class.}\label{fig:summary}
\end{figure}

\section{Contributions}\label{sec:contrib}

This work is a systematic study of the complexity classes $\VVc$, $\VV$, $\MV$, $\SV$, $\VB$, $\MB$, and $\SB$, as well as their constant-time counterparts. Our main contributions are two-fold.

First, we present a complete characterisation of the containment relations between these classes. The definitions of the classes imply the partial order depicted in Figure~\ref{fig:beforeafter}a. For example, classes $\VB$ and $\SV$ are seemingly orthogonal, and it would be natural to assume that neither $\VB \subseteq \SV$ nor $\SV \subseteq \VB$ holds. However, we show that this is not the case. Unexpectedly, the classes form a linear order (see Figure~\ref{fig:beforeafter}b):
\begin{align}
    \SB \subsetneq \MB = \VB \subsetneq \SV = \MV = \VV \subsetneq \VVc. \label{eq:order}
\end{align}
In summary, instead of seven classes that are possibly distinct, we have precisely four distinct classes. These four distinct classes of problems can be concisely characterised as follows, from the strongest to the weakest:
\begin{enumerate}[noitemsep,label=(\arabic*)]
    \item consistent port numbering (class $\VVc$),
    \item no incoming port numbers (class $\SV$ and equivalent),
    \item no outgoing port numbers (class $\VB$ and equivalent),
    \item neither (class $\SB$).
\end{enumerate}
We also show an analogous result for the constant-time versions:
\begin{align}
    \SBl \subsetneq \MBl = \VBl \subsetneq \SVl = \MVl = \VVl \subsetneq \VVcl. \label{eq:orderl}
\end{align}
The main technical achievement here is proving that $\SVl = \MVl$ and $\SV = \MV$. This together with the ideas of a prior work~\cite{astrand10vc-sc} leads to the linear orders \eqref{eq:order} and \eqref{eq:orderl}.

As our second contribution, we identify a novel connection between distributed computational complexity and modal logic. In particular, classes $\VVcl$, $\VVl$, $\MVl$, $\SVl$, $\VBl$, $\MBl$, and $\SBl$ have natural characterisations using certain variants of modal logic. This correspondence allows one to apply tools from the field of modal logic---in particular, bisimulation---to facilitate the proofs of \eqref{eq:order} and~\eqref{eq:orderl}. Conversely, we can lift our results from the field of distributed algorithms to modal logic, by re-interpreting the relations identified in~\eqref{eq:orderl}.

Some of the equivalences between the classes are already known by prior work---in particular, results that are similar to $\MB = \VB$ and $\MV = \VV \subsetneq \VVc$ are implied by e.g., Boldi et al.~\cite{boldi96symmetry} and Yamashita and Kameda~\cite{yamashita99leader}. The main differences between our work and prior work can be summarised as follows.
\begin{enumerate}
    \item All results related to classes $\SV$ and $\SB$ are new. In particular, we are not aware of any prior work that has studied class $\SV$ in this context.
    \item We approach the classification from the perspective of \emph{locality}. We not only prove the equivalences $\MB = \VB$ and $\SV = \MV = \VV$ but also show that in each case the simulation of the stronger model is \emph{efficient}. The nodes do not need to know any global information on the graph in advance (such as an upper bound on the size of the graph), and the nodes do not need to gather any information beyond their constant-radius neighbourhood. Our proofs yield the identical collapses for the constant-time versions of the classes: $\MBl = \VBl$ and $\SVl = \MVl = \VVl$. Similarly, our separation results only rely on problems that can be solved in constant time in one of the classes, without any global information.
    \item The focus on locality also enables us to introduce the connection with \emph{modal logic}. We show how to derive all separations between the complexity classes with bisimulation arguments.
\end{enumerate}
We will discuss related work in more detail in Section~\ref{sec:rel}; see also Tables \ref{tab:prior} and~\ref{tab:prior2}.

\section{Motivation and Related Work}\label{sec:rel}

In this work, we study \emph{deterministic} distributed algorithms in \emph{anonymous} networks---all state transitions are deterministic, all nodes run the same algorithm, and initially each node knows only its own degree. This is a fairly weak model of computation, and traditionally research has focused on stronger models of distributed computing.

\subsection{Stronger Models}

There are two obvious extensions:
\begin{enumerate}
    \item \emph{Networks with unique identifiers}: Initially, all nodes are labelled with $O(\log n)$-bit, globally unique identifiers. With this extension, we arrive at Linial's~\cite{linial92locality} model of computation; Peleg~\cite{peleg00distributed} calls it the $\LOCAL$ model.
    \item \emph{Randomised distributed algorithms}: The nodes have access to a stream of random bits. The state transitions can depend on the random bits.
\end{enumerate}
Both of these extensions lead to a model that is strictly stronger than any of the models studied in this work. The problem of finding a maximal independent set is a good example of a graph problem that separates the weak models from the above extensions. The problem is clearly not in $\VVc$---a cycle with a symmetric port numbering is a simple counterexample---while it is possible to find a maximal independent set \emph{fast} in both of the above models.

\subsection{Port-Numbering Model (\texorpdfstring{$\VVc$}{VVc})}

While most of the attention is on stronger models, one of the weaker models has been studied extensively since the 1980s. Unsurprisingly, it is the strongest of the family, model $\VVc$, and it is commonly known as the \emph{port-numbering model} in the literature.

The study of the port-numbering model was initiated by Angluin~\cite{angluin80local} in 1980. Initially the main focus was on problems that have a \emph{global} nature---problems in which the local output of a node necessarily depends on the global properties of the input. Examples of papers from the first two decades after Angluin's pioneering work include Attiya et al.~\cite{attiya88computing}, Yamashita and Kameda \cite{yamashita96computing, yamashita96computinga, yamashita96computing-functions}, and Boldi and Vigna~\cite{boldi01effective}, who studied global functions, leader election problems, spanning trees, and topological properties.

Based on the earlier work, the study of the port-numbering model may look like a dead end: positive results were rare. However, very recently, distributed algorithms in the port-numbering model have become an increasingly important research topic---and surprisingly, the study of the port-numbering model is now partially motivated by the desire to understand distributed computing in \emph{stronger} models of computation.

The background is in the study of \emph{local algorithms}, i.e., constant-time distributed algorithms~\cite{suomela13survey}. The research direction was initiated by Naor and Stockmeyer~\cite{naor95what} in 1995, and initially it looked like another area where most of the results are negative---after all, it is difficult to imagine a non-trivial graph problem that could be solved in constant time. However, since 2005, we have seen a large number of local algorithms for a wide range of graph problems: these include algorithms for
vertex covers \cite{astrand09vc2apx, astrand10vc-sc, wiese08impact, moscibroda06locality, polishchuk09simple, kuhn05price, kuhn06price},
matchings \cite{floreen10almost-stable, astrand10weakly-coloured},
dominating sets \cite{czygrinow08fast, lenzen10mds, lenzen10what, lenzen11phd},
edge dominating sets~\cite{suomela10eds},
set covers \cite{kuhn05price,kuhn06price,astrand10vc-sc},
semi-matchings~\cite{czygrinow11semimatching},
stable matchings~\cite{floreen10almost-stable}, and
linear programming \cite{kuhn05price, kuhn06price, floreen08local, floreen08tight, floreen09max-min-lp, floreen08approximating, floreen11max-min-lp}.
Naturally, most of these algorithms are related to approximations and special cases, but nevertheless the sheer number of such algorithms is a good demonstration of the unexpected capabilities of local algorithms.

At first sight, constant-time algorithms in stronger models and distributed algorithms in the port-numbering model seem to be orthogonal concepts. However, in many cases a local algorithm is also an algorithm in the port-numbering model. Indeed, a formal connection between local algorithms and the port-numbering model has been recently identified~\cite{goos12local-approximation}.

\subsection{Weaker Models}

As the study of the port-numbering model has been recently revived, now is the right time to ask if it is justified to use $\VVc$ as the standard model in the study of anonymous networks. First, the definition is somewhat arbitrary---it is not obvious that $\VVc$ is the ``right'' class, instead of $\VV$, for example. Second, while the existence of a port numbering is easily justified in the context of wired networks, weaker models such as $\Broadcast$ and $\Set$ seem to make more sense from the perspective of wireless networks.

If we had no positive examples of problems in classes below $\VVc$, there would be little motivation for pursuing further. However, the recent work related to the vertex cover problem~\cite{astrand10vc-sc} calls for further investigation. It turned out that $2$-approximation of vertex cover is a graph problem that is not only in $\VVcl$, but also in $\MBl$---that is, we have a non-trivial graph problem that does not require any access to either outgoing or incoming port numbers. One ingredient of the vertex cover algorithm is the observation that $\MBl = \VBl$, which raises the question of the existence of other similar collapses in the hierarchy of weak models.

We are by no means the first to investigate the weak models. Computation in models that are strictly weaker than the standard port-numbering model has been studied since the 1990s, under various terms---see Table~\ref{tab:prior} for a summary of terminology, and Table~\ref{tab:prior2} for an overview of the main differences in the research directions. Questions related to specific problems, models, and graph families have been studied previously, and indeed many of the techniques and ideas that we use are now standard---this includes the use of symmetry and isomorphisms, local views, covering graphs (lifts) and universal covering graphs, and factors and factorisations. Mayer, Naor, and Stockmeyer \cite{naor95what, mayer95local} made it explicit that the parity of node degrees makes a huge difference in the port-numbering model, and Yamashita and Kameda~\cite{yamashita96computing} discussed factors and factorisations in this context; the underlying graph-theoretic observations can be traced back to as far as Petersen's 1891 work~\cite{petersen1891dietheorie}. Some equivalences and separations between the classes are already known, or at least implicit in prior work---see, in particular, Boldi et al.~\cite{boldi96symmetry} and Yamashita and Kameda~\cite{yamashita99leader}.

\begin{table}
    \newcommand{\mylsep}{\addlinespace[2ex]}
    \centering
    \begin{tabular*}{\columnwidth}{@{\extracolsep{\fill}}llll@{}}
        \toprule
        Algorithm & Problem & Term & References \\
        class     & class   &      &            \\
        \midrule \addlinespace
        $\Vector$ & $\VVc$
            &  port numbering & \cite{angluin80local} \\
            && local edge labelling & \cite{yamashita96computing} \\
            && local orientation & \cite{chalopin06groupings, flocchini03computing} \\
            && orientation & \cite{moran93gap} \\
            && complete port awareness & \cite{boldi96symmetry} \\
            && monoid graph & \cite{norris94classifying-anonymous} \\
            && port-to-port & \cite{yamashita89electing,yamashita99leader} \\
            && port-\`a-port & \cite{chalopin06phd} \\
        \mylsep
        $\Vector$ & $\VV$
            &  input/output port awareness & \cite{boldi96symmetry} \\
        \mylsep
        $\Multiset$ & $\MV$
            &  output port awareness & \cite{boldi96symmetry} \\
            && wireless in input & \cite{boldi97computing} \\
            && mailbox & \cite{boldi97computing} \\
            && port-to-mailbox & \cite{yamashita89electing,yamashita99leader} \\
            && port-\`a-bo\^\i te & \cite{chalopin06phd} \\
        \mylsep
        $\Set$ & $\SV$
            &  --- \\
        \mylsep
        $\Broadcast$ & $\VB$
            &  input port awareness & \cite{boldi96symmetry} \\
            && wireless in output & \cite{boldi97computing} \\
            && broadcast & \cite{yamashita89electing,boldi97computing} \\
            && broadcast-to-port & \cite{yamashita99leader} \\
            && diffusion-\`a-port & \cite{chalopin06phd} \\
        \mylsep
        $\Multiset \cap \Broadcast$ & $\MB$
            &  totalistic & \cite{wolfram83statistical} \\
            && wireless & \cite{diks95anonymous,norris95computing,boldi97computing} \\
            && broadcast-to-mailbox & \cite{yamashita99leader} \\
            && diffusion-\`a-bo\^\i te & \cite{chalopin06phd} \\
            && mailbox-to-mailbox & \cite{yamashita89electing} \\
            && network without colours & \cite{boldi01effective} \\
            && broadcast & \cite{astrand10vc-sc} \\
            && (no name) & \cite{kuhn06complexity} \\
        \mylsep
        $\Set \cap \Broadcast$ & $\SB$
            &  beeping & \cite{cornejo10deploying,afek11beeping} \\
        \addlinespace \bottomrule \addlinespace
    \end{tabular*}
    \caption{Prior work related to the weak models, and a summary of the related terminology. We have identified the closest equivalent in our classification, not necessarily an exact match---see also Table~\ref{tab:prior2} for an overview of the main differences.}\label{tab:prior}
\end{table}

\begin{table}
    \newcommand{\mylsep}{\addlinespace[2ex]}
    \centering
    \begin{tabular*}{\columnwidth}{@{\extracolsep{\fill}}lp{9cm}@{}}
        \toprule
        References & Difference \\
        \midrule \addlinespace
        \cite{norris95computing, boldi96symmetry, boldi01effective, chalopin06groupings, yamashita99leader, yamashita96computing} &
            Focuses on the case of a known topology $G = (V,E)$, a known $|V|$, or a known upper bound on $|V|$.
        \\ \mylsep
        \cite{boldi96symmetry, yamashita99leader} &
            Proves equivalences between the models from a \emph{global} perspective; the simulation overhead can be linear in $|V|$. Our work shows that the equivalences hold also from a \emph{local} perspective; the simulation overhead is bounded by a constant.
        \\ \mylsep
        \cite{norris94classifying-anonymous, norris95computing, boldi97computing} &
            Studies functions that map the local inputs of the nodes to specific local outputs of the nodes. Our work studies graph problems---the local outputs depend on the structure of $G$, not on the local inputs.
        \\ \mylsep
        \cite{norris94classifying-anonymous, norris95computing, yamashita96computinga, boldi01effective} &
            Considers the problem of deciding whether a given problem can be solved in a given graph. In our work, we are interested in the existence of a problem and a graph that separates two models.
        \\ \mylsep
        \cite{yamashita89electing, boldi96symmetry, yamashita96computinga, yamashita99leader, kuhn06complexity, astrand10vc-sc, afek11beeping} &
            Studies individual problems, not classes of problems.
        \\ \mylsep
        \cite{boldi99computing, boldi01effective} &
            Provides general results, but does not study the implications from the perspective of the weak models and their relative strength.
        \\ \mylsep
        \cite{angluin80local, attiya88computing, mayer95local, yamashita96computing, yamashita96computinga, chalopin06groupings} &
            Does not consider models that are weaker than the port-numbering model.
        \\ \mylsep
        \cite{wolfram83statistical, attiya88computing, diks95anonymous, kuhn06complexity, flocchini03computing, moran93gap} &
            Assumes a specific network structure (cycle, grid, etc.), or auxiliary information in local inputs.
        \\ \mylsep
        \cite{afek11beeping, emek12stone} &
            Studies randomised, asynchronous algorithms.
        \\
        \addlinespace \bottomrule \addlinespace
    \end{tabular*}
    \caption{Main differences in the problem setting between this work and selected prior work.}\label{tab:prior2}
\end{table}

However, it seems that a comprehensive classification of the weak models from the perspective of solvable graph problems has been lacking. Our main contribution is putting all pieces together in order to provide a complete characterisation of the relations between the weak models and the complexity classes associated with them.

We also advocate a new perspective for studying the weak models---the connections with modal logic can be used to complement the traditional graph-theoretic approaches. In particular, bisimulation is a convenient tool that complements the closely related graph-theoretic concepts of covering graphs and fibrations.

\subsection{Local Inputs}\label{ssec:localinput}

In this article we study graph problems associated with simple undirected graphs of the type $(V,E)$. It would also be worthwhile to study structures of the type $(V,E,f)$, where $f\colon V\rightarrow\mathbb{N}$ is function encoding a \emph{local input} $f(u)$ associated with each node $u\in V$. The related notion of a state machine would be the same as in Section \ref{execution}, with the additional property that the initial state $x_0(u)$ of a machine at a node $u$ would depend on the local information $f(u)$ in addition to the degree of $u$.

While we will not study the effects of local inputs, it is worth noticing that the classification given by \eqref{eq:order} and \eqref{eq:orderl} extends immediately to the context with local information---in particular, a separation with unlabelled graphs implies a separation in the more general case of labelled graphs.

As long as each node knows its own degree, local inputs do not seem to add anything interesting to the classification of weak models of distributed computing---a uniformly finite amount of local information could be encoded in the topological information of the graph. However, if we studied models that are strictly weaker than $\SB$ (for example, model $\SBo$ that we briefly mentioned in Remark~\ref{rem:sbo}), local inputs would be necessary in order to arrive at non-trivial results.

\subsection{Distributed Algorithms and Modal Logic}

Modal logic (see Section~\ref{sec:logic}) has, of course, been applied previously in the context of distributed systems. For example, in their seminal paper, Halpern and Moses~\cite{halpern90knowledge} use modal logic to model epistemic phenomena in distributed systems. A distributed system $S$ gives rise to a \emph{Kripke model} (see Section~\ref{ssec:logics}), whose set $W$ of domain points corresponds to the set of partial runs of $S$, that is, finite sequences of global states of $S$. For each processor $i$ of $S$, there is an \emph{accessibility relation} $R_i$ such that $(v,w)\in R_i$ if and only if $v$ and $w$ are \emph{indistinguishable} from the point of view of processor $i$. This framework suits well for epistemic considerations.

In traditional modal approaches, the domain elements of a Kripke model correspond to possible states of a distributed computation process. Our perspective is a radical departure from this approach. In our framework, a distributed system is---essentially---a Kripke model, where the domain points are processors and the accessibility relations are communication channels. While such an interpretation is of course always possible, it turns out to be particularly helpful in the study of weak models of distributed computing. With this interpretation, for example, a local algorithm in $\Set \cap \Broadcast$ corresponds to a formula of modal logic, while a local algorithm in $\Multiset \cap \Broadcast$ corresponds to a formula of \emph{graded} modal logic---local algorithms are exactly as expressive as such formulas, and the running time of an algorithm equals the modal depth of a formula. Standard techniques from the field of modal logic can be directly applied in the study of distributed algorithms, and conversely our classification of the weak models of distributed computing can be rephrased as a result that characterises the expressibility of modal logics in certain classes of Kripke models.

\section{Connections with Modal Logic}\label{sec:logic}

In this section, we show how to characterise each of the classes $\SBl$, $\MBl$, $\VBl$, $\SVl$, $\MVl$, $\VVl$, and $\VVcl$ by a corresponding modal logic, in the spirit of descriptive complexity theory (see Immerman~\cite{immerman99descriptive}). We show that for each class there is a modal logic that is equally expressive: for any graph problem in the class there is a formula in the modal logic that defines a solution of the graph problem; conversely, any formula in the modal logic defines a solution of some graph problem in the class.

\subsection{Logics \texorpdfstring{$\ML$}{ML}, \texorpdfstring{$\GML$}{GML}, \texorpdfstring{$\MML$}{MML}, and \texorpdfstring{$\GMML$}{GMML}}\label{ssec:logics}

Our characterisation uses \emph{basic modal logic} $\ML$, \emph{graded modal logic} $\GML$, \emph{multimodal logic} $\MML$, and \emph{graded multimodal logic} $\GMML$---see, e.g., 
Blackburn, de Rijke, and Venema
\cite{blackburn01modal} or
Blackburn, van Benthem, and Wolter
\cite{blackburn07handbook} for further details on modal logic.

\emph{Basic modal logic}, $\ML$, is obtained by extending propositional logic by a single 
(unary) modal operator $\Diamond$. 
More precisely, if $\Phi$ is a finite set of proposition symbols, then 
the set of $\ML(\Phi)$-formulas is given by the following grammar:
\[
    \varphi \, := \, q
        \,\mid\, (\varphi\wedge\varphi)
        \,\mid\, \neg\varphi
        \,\mid\, \Diamond\varphi,
        \quad \text{where } q\in \Phi.
\]

The semantics of $\ML$ is defined on Kripke models. A \emph{Kripke model} for the set $\Phi$
of proposition symbols is a tuple $\KM=(W,R,\tau)$, where $W$ is a nonempty set
of \emph{states} (or \emph{possible worlds}), $R\subseteq W^2$ is a binary relation on $W$ 
(\emph{accessibility relation}), and $\tau$ is a \emph{valuation} function 
$\tau \colon \Phi \to \pow(W)$.

The truth of an $\ML(\Phi)$-formula $\varphi$ in a state $v\in W$ of a Kripke model 
$\KM=(W,R,\tau)$ is defined recursively as follows:
\begin{alignat*}{2}
    \KM,v&\models q
    & \quad \text{iff} \quad &
    v\in \tau(q),\ \text{ for each }q\in \Phi,
    \\
    \KM,v&\models (\varphi\wedge\vartheta)
    & \quad \text{iff} \quad &
    \KM,v\models \varphi \text{ and } \KM,v\models \vartheta,
    \\
    \KM,v&\models \neg\varphi
    & \quad \text{iff} \quad &
    \KM,v\nmodels \varphi,
    \\
    \KM,v&\models \Diamond\varphi
    & \quad \text{iff} \quad &
    \KM,w\models \varphi
    \text{ for some $w\in W$ such that\ $(v,w)\in R$}.
\end{alignat*}
Usually in modal logic one defines the abbreviations $(\varphi\vee\vartheta):=\neg(\neg\varphi\wedge\neg\vartheta)$ and $\Box\varphi:=\neg\Diamond\neg\varphi$.

Classical modal logic has its roots in the philosophical analysis of the 
notion of possibility.
In classical modal logic, a modal formula $\Diamond \varphi$ is interpreted to mean that
it is \emph{possible} that $\varphi$ holds. The set $W$ of a 
Kripke model $\KM=(W,R,\tau)$ is a collection of possible worlds $v$,
or possible states of affairs. The relation $R$ connects a possible world $v$ to
exactly those worlds that can be considered to be---in one sense or another---\emph{possible}
states of affairs, when the \emph{actual} state of affairs is in fact $v$. The semantics of 
the formula $\Diamond\varphi$ reflects this idea; $\Diamond\varphi$ is true in $v$ if and only if there is a
possible state of affairs $w$ accessible from $v$ via $R$ such that $\varphi$  is true in $w$.

Modern systems of modal logic often have very little to do with the original philosophical 
motivations of the field. The reason is that modal logic and Kripke semantics seem to adapt rather well
to the requirements of a wide range of different kinds of applications in computer science
and various other fields. Our use of modal logic in this article is an example of such an
adaptation.

One of the features of basic modal logic is that it is unable to count: there is no mechanism in $\ML$ for
separating states $v$ of Kripke models based only on the number of $R$-successors of $v$.
The most direct way to overcome this defect is to add counting to the modalities. The syntax of 
\emph{graded modal logic}~\cite{fine72in}, $\GML$, extends the syntax of $\ML$ with the rules 
$\Diamond_{\geq k}\varphi$, where $k\in\mathbb{N}$. The semantics of
these graded modalities $\Diamond_{\geq k}$ is the following:
\[
    \KM,v\models \Diamond_{\geq k}\varphi
    \quad \text{iff} \quad
    \bigl\lvert\{w\in W: (v,w)\in R \text{ and } \KM,w\models \varphi \}\bigr\rvert \,\geq\, k.
\]

Up to this point we have considered modal logics with only one modality~$\Diamond$. 
\emph{Multimodal logic}, $\MML$, is the natural generalisation of $\ML$ that allows an arbitrary (finite) number of 
modalities. The modalities are usually written as $\langle \alpha\rangle$, where $\alpha\in I$ for some
index set $I$.
Given the set $I$ and a finite set $\Phi$ of proposition symbols, the set of $\MML(I,\Phi)$-formulas
is defined by the following grammar:
\[
    \varphi \, := \, q
        \,\mid\, (\varphi\wedge\varphi)
        \,\mid\, \neg\varphi
        \,\mid\, \langle \alpha\rangle\varphi,
        \quad \text{where } q\in \Phi \text{ and } \alpha \in I.
\]

The Kripke models corresponding to the multimodal language $\MML(I,\Phi)$
are of the form $\KM=(W,(R_\alpha)_{\alpha\in I},\tau)$, where $R_\alpha\subseteq W^2$
for each $\alpha\in I$, and $\tau$ is a function $\tau \colon \Phi \to \pow(W)$.

The truth  definition of $\MML(I,\Phi)$ is the same as the truth definition of $\ML$ for Boolean connectives and 
atomic formulas. For diamond formulas $\langle \alpha \rangle\varphi$ the semantics are given by the condition
\[
  \KM,v\models \langle \alpha \rangle\varphi
  \quad \text{iff} \quad
  \KM,w\models \varphi
  \text{ for some $w\in W$ s.t.\ $(v,w)\in R_\alpha$}.
\]
We can naturally extend $\MML$ by graded modalities $\langle\alpha \rangle_{\geq k}$ for each $\alpha\in I$ and 
$k\in \N$ and obtain \emph{graded multimodal logic} $\GMML(I,\Phi)$.

If the index set $I$ contains only one element $\alpha$, then  $\MML(I,\Phi)$ can be identified with $\ML(\Phi)$
simply by replacing $\langle\alpha\rangle$ with $\Diamond$. Similarly, $\GMML(\{\alpha\},\Phi)$ is identified with 
$\GML(\Phi)$.

Let $\Lo$ be a modal logic and $\varphi$ an $\Lo(I,\Phi)$-formula. The \emph{modal depth} of $\varphi$, denoted by 
$\md(\varphi)$, is defined recursively as follows:
\begin{align*}
    \md(q) &= 0 \text{ for } q\in\Phi, \\
    \md(\varphi\wedge\vartheta) &= \max\{\md(\varphi),\md(\vartheta)\}, \\
    \md(\lnot\varphi) &= \md(\varphi), \\
    \md(\langle\alpha\rangle\varphi) &= \md(\varphi)+1 \text{ for } \alpha\in I.
\end{align*}
Thus, $\md(\varphi)$ is the largest number of nested modalities in $\varphi$. 

Given a modal logic $\Lo$ and a Kripke model $\KM$ for $\Lo$,
each $\Lo$-formula $\varphi$ \emph{defines} a subset $\{v\in W\mid \KM,v\models \varphi\}$ of the set of states in $\KM$; 
this set is denoted by $\lVert\varphi \rVert^\KM$.

\subsection{Bisimulation and Definability in Modal Logic}

We will now define one of the most important concepts in modal logic, bisimulation. Bisimulation was first defined in the context of modal logic by van Benthem~\cite{vBenthem76}, who calls it a \emph{p-relation}. Bisimulation was also discovered independently in a variety of other fields. See Sangiorgi~\cite{Sangiorgi09} for the history and development of the notion.

The objective of bisimulation is to characterise definability in the corresponding modal logics, 
so that if two states $w$ and $w'$ are bisimilar they 
cannot be separated by any formula of the corresponding logic.
 Bisimulation can be defined in a canonical way
for each of the logics $\ML$, $\GML$, $\MML$, and $\GMML$.

Bisimulation for $\MML$ is defined as follows.
Let
\begin{align*}
    \KM &= \bigl(W,\,(R_\alpha)_{\alpha\in I},\,\tau\bigr), \\
    \KM' &= \bigl(W',\,(R'_\alpha)_{\alpha\in I},\,\tau'\bigr)
\end{align*}
be Kripke models for a set $\Phi$ of proposition symbols. A nonempty relation $Z\subseteq W\times W'$ 
is a \emph{bisimulation} between $\KM$ and $\KM'$ if the following conditions hold.
\begin{enumerate}[label=(B\arabic*)]
\item\label{b1} If $(v,v')\in Z$, then $v\in\tau(q)$ iff $v'\in\tau'(q)$ for all $q\in\Phi$.
\item\label{b2} If $(v,v')\in Z$ and $(v,w)\in R_\alpha$ for some $\alpha\in I$, then there is a $w'\in W'$ such that $(v',w')\in R'_\alpha$ and $(w,w')\in Z$. 
\item\label{b3} If $(v,v')\in Z$ and $(v',w')\in R'_\alpha$ for some $\alpha\in I$, then there is a $w\in W$ such that $(v,w)\in R_\alpha$ and $(w,w')\in Z$.
\end{enumerate}
If there is a bisimulation $Z$ such that $(v,v')\in Z$, we say that $v$ and $v'$ are \emph{bisimilar}.

For the basic modal logic $\ML$, bisimulation is defined in the same way just by replacing the relations 
$R_\alpha$, $\alpha\in I$, in conditions \ref{b2} and \ref{b3} with the single relation $R$. 

In the case of the graded modal logic $\GML$, we use the notion
of \emph{graded bisimulation}: a nonempty relation $Z\subseteq W\times W'$ 
is a graded bisimulation between $\KM=(W,R,\tau)$ and $\KM'=(W',R',\tau')$ if it satisfies condition 
\ref{b1} and the following modifications of \ref{b2} and \ref{b3}; we use 
 the notation $R(v)=\{w\in W:\nobreak (v,w)\in R\}$.
\begin{enumerate}[label=(B\arabic*$^*$),start=2]
\item If $(v,v')\in Z$ and $X\subseteq R(v)$, then there is a set $X'\subseteq R'(v')$ such that $|X'|=|X|$ and
for each $w'\in X'$ there is a $w\in X$ with $(w,w')\in Z$. 
\item If $(v,v')\in Z$ and $X'\subseteq R'(v')$, then there is a set $X\subseteq R(v)$ 
such that $|X|=|X'|$ and
for each $w\in X$ there is a $w'\in X'$ with $(w,w')\in Z$. 
\end{enumerate}
We say that $v$ and $v'$ are \emph{g-bisimilar} if there is a graded bisimulation $Z$ 
such that $(v,v')\in Z$.

The definition of graded bisimulation for $\GMML$ is the obvious generalisation of the definition above
to the case of several relations $R_\alpha$ instead of a single relation $R$. 

The notion of graded bisimulation was first formulated by de Rijke~\cite{dRijke00}. Our 
definition follows the formulation of Conradie~\cite{Conradie02}. We state next the main result 
concerning bisimulation. For the proof of Fact~\ref{fact:bisim}a, we refer to Blackburn et al.~\cite{blackburn01modal}.
The proof of Fact~\ref{fact:bisim}b can be found in Conradie~\cite{Conradie02}.

\begin{fact}\label{fact:bisim}
{\normalfont(a)}~Let $\Lo$ be $\ML$ or $\MML$, and
let $\KM$ and $\KM'$ be Kripke models, $v\in W$ and $v'\in W'$. If $v$ and $v'$ are bisimilar,
then for all $\Lo$-formulas $\varphi$
\[
\KM,v\models \varphi \text{ iff } \KM',v'\models \varphi.
\]

{\normalfont(b)}~Let $\Lo$ be $\GML$ or $\GMML$, and
let $\KM$ and $\KM'$ be Kripke models, $v\in W$ and $v'\in W'$. If $v$ and $v'$ are g-bisimilar,
then for all $\Lo$-formulas $\varphi$
\[
\KM,v\models \varphi \text{ iff } \KM',v'\models \varphi.
\]
\end{fact}

In what follows, we will develop a connection between modal logic and weak models of distributed computing. Informally, the states of a Kripke model will correspond to the nodes of a distributed system, and bisimilar states will correspond to nodes that are unable to distinguish their neighbourhoods, no matter which distributed algorithm we use. With the help of this connection, we can then use bisimulation in Section~\ref{sec:separ} to prove separations of problem classes.

\subsection{Characterising Constant-Time Classes by Modal Logics}\label{ssec:char-logic}

There is a natural correspondence between the framework for distributed computing defined in this paper
and the logics $\ML$, $\GML$, $\MML$, and $\GMML$. 
For any input graph $G$ and port numbering $p$ of $G$, the pair $(G,p)$ can be transformed
into a Kripke model $\KM(G,p)=(W,(R_\alpha)_{\alpha\in I},\tau)$ in a canonical way.
Given a local algorithm $\aA$, its execution can then be simulated by a modal 
formula $\varphi$. The crucial idea is that the truth condition for a diamond formula 
$\langle \alpha\rangle\psi$ is interpreted as communication between the nodes:
\[
\begin{split}
   \KM,v\models\langle \alpha\rangle\psi
   \quad\text{iff}\quad
   &\text{$v$ receives the message ``$\psi$ is true"} \\[-2pt]
   &\text{from some $u$ such that $(v,u)\in R_\alpha$.}
\end{split}
\]
Conversely, for any modal formula $\varphi$, there is a local algorithm $\aA$ that can evaluate the truth of $\varphi$ in the Kripke model $\KM(G,p)$.

The general idea of the correspondence between modal logic and distributed algorithms is described in Table~\ref{tab:logic}. We will assume that $\aA$ produces a one-bit output, i.e., $Y = \{0,1\}$; other cases can be handled by defining a separate formula for each output bit.

\begin{table}
    \centering
    \begin{tabular}{@{}lll@{}}
        \toprule
        Modal logic && Distributed algorithms \\
        \midrule
        \multirow{2}{*}{Kripke model $\KM = (W,(R_\alpha)_{\alpha\in I},\tau)$}
            & \multirow{2}{*}{\huge \{}
            & input graph $G = (V,E)$ \\
           && port numbering $p$ \\
        \addlinespace
        states $W$ && nodes $V$ \\
        relations $R_\alpha$, $\alpha\in I$ && edges $E$ and port numbering $p$ \\
        \addlinespace
        valuation $\tau$ & \multirow{2}{*}{\huge \}} & \multirow{2}{*}{node degrees (initial state)} \\
        proposition symbols $q_1, q_2, \dotsc$ \\
        \midrule
        formula $\varphi$ && algorithm $\aA$ \\
        formula $\varphi$ is true in state $v$ && algorithm $\aA$ outputs $1$ in node $v$ \\
        modal depth of $\varphi$ && running time of $\aA$ \\
        \bottomrule
    \end{tabular}
    \caption{Correspondence between modal logic and distributed algorithms.}\label{tab:logic}
\end{table}

We start by defining the Kripke models $\KM(G,p)$. There are in fact four different versions of
$\KM(G,p)$, reflecting the fact that algorithms in the lower classes do not use all the information
encoded in the port numbering. 
Let $G=(V,E)\in\F(\Delta)$, and let $p$ be
a port numbering of $G$. The accessibility relations used in the different versions of 
$\KM(G,p)$ are the following; see Figure~\ref{fig:relations} for illustrations:
\[
    R_{(i,j)}=\{(u,v)\in V\times V: p((v,j))=(u,i)\}  \quad \text{for each pair } (i,j)\in [\Delta]\times [\Delta].
\]
Given $\Delta$, these relations together with the vertex set $V$ contain the same information as the pair $(G,p)$: graph $G$ and
port numbering $p$ can be reconstructed from the pair
\[
    \bigl(V,(R_{(i,j)})_{(i,j)\in [\Delta]\times [\Delta]}\bigr).
\]
Since algorithms in classes below $\Vector$ have access to a restricted part of the information in $p$, 
we need alternative accessibility relations with corresponding restrictions on their information about~$p$:
\begin{alignat*}{2}
    R_{(i,*)} &= \bigcup_{j\in [\Delta]} R_{(i,j)}
    & \quad & \text{for each } i\in [\Delta], \\[3pt]
    R_{(*,j)} &=\bigcup_{i\in [\Delta]} R_{(i,j)}
    & \quad & \text{for each } j\in [\Delta], \\[3pt]
    R_{(*,*)} &=\bigcup_{(i,j)\in [\Delta]\times [\Delta]} R_{(i,j)}.
\end{alignat*}
Note that $R_{(*,*)}=\{(u,v): \{u,v\}\in E\}$ is the edge set $E$ interpreted as a symmetric relation.

\begin{figure}[t]
    \centering
    \includegraphics[page=\PRelations]{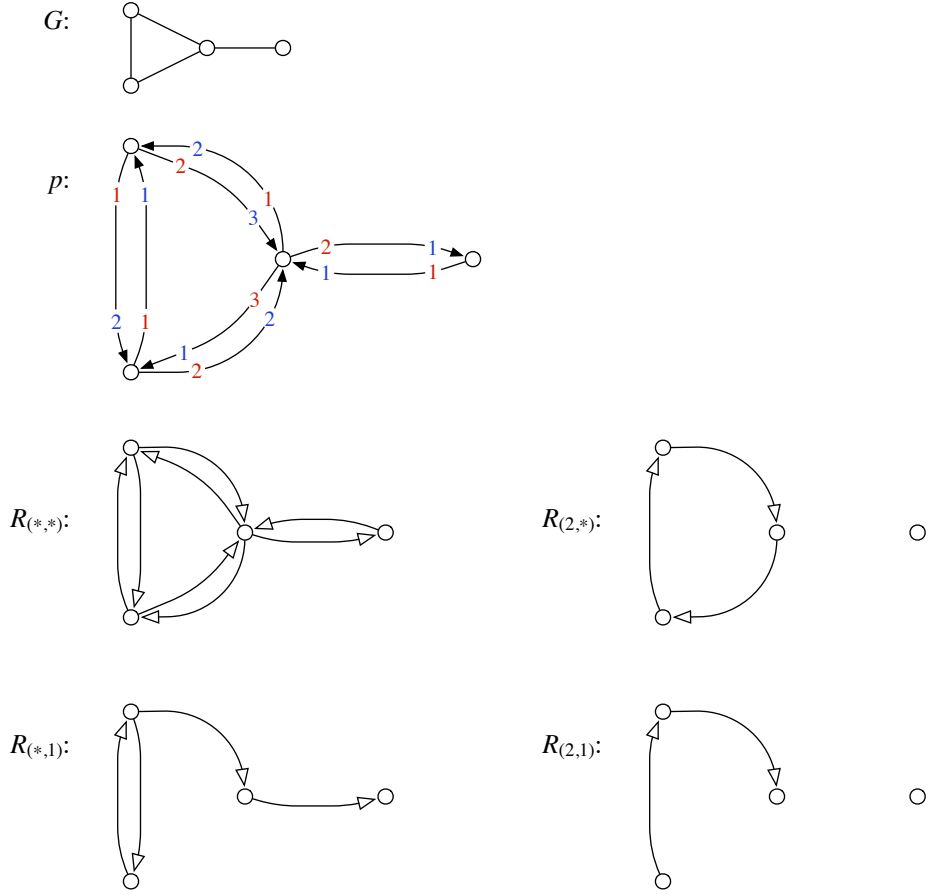}
    \caption{Relations $R_{(i,j)}$---note the directions of the arrows.}\label{fig:relations}
\end{figure}

In addition to the accessibility relations, we encode the local information on the degrees of vertices 
into a valuation
$\tau \colon \Phi_\Delta \to \pow(V)$, where $\Phi_\Delta=\{q_i: i\in[\Delta]\}$. The valuation $\tau$ is given by
\[
    \tau(q_i)=\{v\in V: \deg(v)=i\}.
\]

The four versions of a Kripke model corresponding to graph $G$ and port numbering $p$ are now defined as follows:
\begin{alignat*}{2}
    \KVV(G,p) &= (V, (R_\alpha)_{\alpha\in I^{\Delta}_{+,+}},\tau),
    \quad & \text{where} \quad 
    I^{\Delta}_{+,+} &= [\Delta]\times [\Delta],
    \\
    \KMV(G,p) &= (V, (R_\alpha)_{\alpha\in I^{\Delta}_{-,+}},\tau),
    \quad & \text{where} \quad 
    I^{\Delta}_{-,+} &= \{*\}\times [\Delta],
    \\
    \KVB(G,p) &= (V, (R_\alpha)_{\alpha\in I^{\Delta}_{+,-}},\tau),
    \quad & \text{where} \quad 
    I^{\Delta}_{+,-} &= [\Delta]\times\{*\},
    \\
    \KMB(G,p) &= (V, (R_\alpha)_{\alpha\in I^{\Delta}_{-,-}},\tau),
    \quad & \text{where} \quad 
    I^{\Delta}_{-,-} &= \{(*,*)\}.
\end{alignat*}
For all $a,b\in \{-,+\}$, we denote the class of all Kripke models of the form $\KM_{a,b}(G,p)$ by $\CM_{a,b}$.
Furthermore, we denote by $\CVVc$ the subclass of $\CVV$ consisting of the models $\KVV(G,p)$,
where $p$ is a consistent port numbering of $G$.

In order to give a precise formulation to the correspondence between modal logics and the constant-time classes of graph problems, we define the concept of modal formulas solving graph problems. Without loss of generality, we consider here only problems $\Pi$ such that the solutions $S\in\Pi(G)$ are functions $V\to \{0,1\}$, or equivalently, subsets of $V$. This is a natural restriction, since a modal formula $\psi$ defines a subset
\[
    \lVert\psi\rVert^{\KM_{a,b}(G,p)} := \{v\in V\mid \KM_{a,b}(G,p),v\models \psi\}
\]
of the vertex set~$V$. Other cases can be handled by using tuples of formulas.

Let $a,b\in \{-,+\}$, and let $\Psi=(\psi_1,\psi_2,\ldots)$ be a sequence of modal formulas such that 
$\psi_\Delta$ is in the signature $(I^{\Delta}_{a,b},\Phi_\Delta)$ for each $\Delta\in\N$. Then $\Psi$  
\emph{defines a solution} for a graph 
problem $\Pi$ on the class $\CM_{a,b}$ if the following condition holds:
\begin{itemize}
\item For all $\Delta\in\N$, all $G\in\F(\Delta)$, and all port numberings $p$ of $G$,
the subset $\lVert\psi_\Delta\rVert^{\KM_{a,b}(G,p)}$ defined by the formula $\psi_\Delta$ in the model
$\KM_{a,b}(G,p)$ is in the set $\Pi(G)$. 
\end{itemize}
Furthermore, the sequence $\Psi$ defines a solution for $\Pi$ on the class $\CVVc$, if the condition above
with $a=b=+$ holds for all \emph{consistent} port numberings $p$. 

Note that any sequence $\Psi=(\psi_1,\psi_2,\ldots)$ of modal formulas as above 
gives rise to a canonical graph problem $\Pi_\Psi$ that it defines a solution for: 
for each graph
$G$, the solution set $\Pi_\Psi(G)$  simply consists of the sets 
$\lVert\psi_\Delta\rVert^{\KM_{a,b}(G,p)}$ where $G\in\F(\Delta)$ and
$p$ ranges over the (consistent) port numberings of $G$.

Let $\Lo$ be a modal logic, let $a,b\in\{-,+\}$, and let $C$ be a class of graph problems.
We say that $\Lo$ \emph{is contained in} $C$ on $\CM_{a,b}$, in symbols $\Lo\le C$
on $\CM_{a,b}$, if the following condition holds:
\begin{itemize}
\item If $\Psi=(\psi_1,\psi_2,\ldots)$ is a sequence of formulas such that 
$\psi_\Delta\in\Lo(I^{\Delta}_{a,b},\Phi_\Delta)$ for all $\Delta\in\N$, then $\Pi_\Psi\in C$.
\end{itemize}
Furthermore, we say that $\Lo$ \emph{simulates} $C$ on $\CM_{a,b}$, in symbols $C\le\Lo$
on $\CM_{a,b}$, if the following condition holds:
\begin{itemize}
\item For every graph problem $\Pi\in C$ there is a sequence $\Psi=(\psi_1,\psi_2, \ldots)$
of formulas such that $\psi_\Delta\in\Lo(I^{\Delta}_{a,b},\Phi_\Delta)$ for all $\Delta\in\N$, 
which defines a solution for $\Pi$ on $\CM_{a,b}$.
\end{itemize}
Finally, we say that $\Lo$ \emph{captures} $C$ on $\CM_{a,b}$ if both $\Lo\le C$
and $C\le\Lo$ on $\CM_{a,b}$.

The notions of $\Lo$ being contained in $C$ on $\CVVc$, $\Lo$ 
simulating $C$ on $\CVVc$, and $\Lo$ capturing $C$ on $\CVVc$
are defined similarly with the obvious restriction to consistent port numberings.

The main result of this section is that the constant-time version of each of the 
classes $\VVc$, $\VV$, $\MV$, $\SV$, $\VB$,
$\MB$, and $\SB$ is captured by one of the modal logics $\MML$, $\ML$, $\GMML$, and $\GML$ on an appropriate class $\CM_{a,b}$.

\begin{theorem}\label{thm:logic}
{\normalfont(a)} $\MML$ captures $\VVcl$ on $\CVVc$.

{\normalfont(b)} $\MML$ captures $\VVl$ on $\CVV$.

{\normalfont(c)} $\GMML$ captures $\MVl$ on $\CMV$.

{\normalfont(d)} $\MML$ captures $\SVl$ on $\CMV$.

{\normalfont(e)} $\MML$ captures $\VBl$ on $\CVB$.

{\normalfont(f)} $\GML$ captures $\MBl$ on $\CMB$.

{\normalfont(g)} $\ML$ captures $\SBl$ on $\CMB$.
\end{theorem}

\paragraph{Proof of Theorem~\ref{thm:logic}: Overview.}

Note first that (a) follows directly from (b) by restricting to consistent port numberings.
Furthermore, the only difference between  $\GMML$ and $\MML$ is the ability to
count the number of neighbours satisfying a formula, which corresponds 
in a natural way to the difference between algorithms in $\Multiset$ and $\Set$.
Hence, we omit the proof of (d), as it is obtained from the proof of (c) by minor modifications.
Similarly, the proof of (g) is a minor modification of the proof of (f), so we omit it, too. 

Thus, we are left with the task of proving claims (b), (c), (e) and (f).
The structure of the proofs of all these claims is the same---there are 
differences only in technical details.
Hence, we divide the proof in four parts as follows. 
\begin{enumerate}[label=\arabic*.]
\item We prove the first half of (b): $\MML\le\VVl$ on $\CVV$. 
\item We describe the changes in part~1 needed for proving the first halves of (c), (e) and (f). 
\item We prove the second half of (b): $\VVl\le\MML$ on $\CVV$. 
\item We describe the changes in part~3 needed for proving the second halves of (c), (e) and (f). 
\end{enumerate}

\paragraph{Proof of Theorem~\ref{thm:logic}, Part 1.}

Assume that $\Psi=(\psi_1,\psi_2,\ldots)$ is a sequence of formulas with 
$\psi_\Delta\in\MML(I^{\Delta}_{+,+},\Phi_\Delta)$ for each $\Delta\in\N$. 
We give for each $\Delta\in\N$ a local algorithm $\aA_\Delta\in\Vector$ that simulates
the recursive evaluation of the truth of $\psi_\Delta$ on a Kripke model $\KVV(G,p)$. 

Let $\Sigma$ be the set of all subformulas of $\psi_\Delta$,
and let $D_j$, $j\in [\Delta]$, be the subset of $\Sigma$ consisting of
subformulas $\eta$ such that $\langle\alpha \rangle\eta\in \Sigma$, where $\alpha=(i,j)$ for some $i\in[\Delta]$.
The set of stopping states, intermediate states, and messages of the algorithm $\aA_\Delta$ (see
Section~\ref{ssec:distalg}) are defined as follows:
\begin{align*}
    Y &:= \{0,1\}, \\
    Z &:= \bigl\{f: f\text{ is a function $\Sigma\to \{0,1,U\}$}\bigr\}, \\
    M &:= \bigcup_{j\in[\Delta]}\bigl\{h: h\text{ is a function } D_j\to \{0,1,U\}\times\{j\}\bigr\}\cup \{m_0\}.
\end{align*}
The idea behind these choices is that before stopping, the state $x_t(v)$ of the computation of $\aA_\Delta$
on a node $v$ of an input $(G,p)$ encodes the truth value of each subformula of $\psi_\Delta$
with modal depth at most $t$; for subformulas with modal depth greater than $t$, the
state $x_t(v)$ gives the value $U$ (undefined).
In other words, our aim is to make sure that at each step $t$ of the computation, $x_t(v)=f$, where $f\in Z$ is the function defined by
\[
   f(\eta)=  \begin{cases} 0, \text{ if $\md(\eta)\le t$ and }\KVV(G,p),v\nmodels\eta \\ 
   1,  \text{ if $\md(\eta)\le t$ and }  \KVV(G,p),v\models\eta  \\ 
   U, \text{ if $\md(\eta)> t$} \end{cases}
\]
for each $\eta\in \Sigma$.
First, we define the function $z_0\colon [\Delta]\to Y\cup Z$ that gives the initial state $x_0(v)=z_0(\deg(v))$ of each node~$v$.
For each $i\in[\Delta]$, we set $z_0(i) = g$, where $g$ is the function defined recursively as follows:
\begin{alignat*}{2}
    g(\eta)&=1
    &\text{for }\eta&= q_i\in \Sigma,
    \\[1ex]
    g(\eta)&=0
    &\text{for }\eta&= q_j\in \Sigma,\\ && j&\in [\Delta]\setminus\{i\},
    \\[1ex]
    g(\eta)&=\begin{cases}
        0, \text{ if $0\in\{g(\vartheta),g(\gamma)\}\subseteq\{0,1\}$} \\
        1, \text{ if $\{g(\vartheta),g(\gamma)\}=\{1\}$} \\ 
        U, \text{ if $U\in\{g(\vartheta),g(\gamma)\}$}
    \end{cases}
    &\text{for }\eta&=(\vartheta\land\gamma)\in \Sigma,
    \\[1ex]
    g(\eta)&=\begin{cases}
        0, \text{ if $g(\vartheta)=1$} \\ 
        1, \text{ if $g(\vartheta)=0$} \\ 
        U, \text{ if $g(\vartheta)=U$}
    \end{cases}
    &\text{for }\eta&=\lnot\vartheta\in \Sigma,
    \\[1ex]
    g(\eta)&=U
    &\text{for }\eta&=\langle\alpha\rangle\vartheta \in \Sigma.
\end{alignat*}

If a node $v$ of the input graph $G$ is in the state $x_t(v)=f\in Z$ at time step~$t$, then the
message $\mu(f,j)$ it sends to its port $j\in [\deg(v)]$ at step $t+1$ is obtained from the restriction 
of $f$ to the
set $D_j$ by adding $j$ as a marker: that is, $\mu(f,j)$ is the function 
$h\colon D_j\to \{0,1,U\}\times\{j\}$ such that $h(\eta)=(f(\eta),j)$ for all $\eta\in D_j$. 

Finally, the state transition function $\delta$ of $\aA_\Delta$ is described as follows.
Assume that the state of a node $v$ at time $t$ is $x_t(v)=f\in Z$, and the vector of messages
it receives at time $t+1$ from the ports is $\vec{a}_{t+1}(v)=(h_1,\ldots,h_\Delta)$. If
$f(\psi_\Delta)=U$, then $x_{t+1}(v)$ is the function $g\in Z$ defined as follows: 
\begin{enumerate}
\item For each $\eta\in \Sigma$ with $f(\eta)\ne U$, we set $g(\eta)=f(\eta)$.
\item For each $\eta\in \Sigma$ with $f(\eta)=U$, we define $g(\eta)$ by the following recursion:
\begin{alignat*}{2}
    \tag{$\delta_\land$}
    g(\eta)&=\begin{cases}
        0, \text{ if $0\in\{g(\vartheta),g(\gamma)\}\subseteq\{0,1\}$} \\ 
        1, \text{ if $\{g(\vartheta),g(\gamma)\}=\{1\}$} \\ 
        U, \text{ if $U\in\{g(\vartheta),g(\gamma)\}$}
    \end{cases}
    &\text{for }\eta&=(\vartheta\land\gamma)\in \Sigma,
    \\[1ex]
    \tag{$\delta_\lnot$}
    g(\eta)&=\begin{cases}
        0, \text{ if $g(\vartheta)=1$} \\ 
        1, \text{ if $g(\vartheta)=0$} \\ 
        U, \text{ if $g(\vartheta)=U$}
    \end{cases}
    &\text{for }\eta&=\lnot\vartheta\in \Sigma,
    \\[1ex]
    \tag{$\delta_\Diamond$}
    g(\eta)&=\begin{cases}
        0, \text{ if $f(\vartheta) \ne U$ and $h_i(\vartheta)\ne(1,j)$} \\
        1, \text{ if $f(\vartheta) \ne U$ and $h_i(\vartheta)=(1,j)$} \\ 
        U, \text{ if $f(\vartheta) = U$}
    \end{cases}
    &\text{for }\eta&=\langle (i,j)\rangle\vartheta\in \Sigma.
\end{alignat*}
For convenience, we interpret $m_0$ as a function with $m_0(\vartheta)=(0,1)$ for each subformula $\vartheta$.
\end{enumerate}
On the other hand, if $f(\psi_\Delta)\ne U$, we let $x_{t+1}(v)=f(\psi_\Delta)\in Y$.

It is now straightforward to prove by induction on modal depth that the following holds for any 
input graph $G\in\F(\Delta)$, port numbering $p$ of $G$, and node $v$ of $G$:
\begin{itemize}
\item If $\eta\in \Sigma$, $\md(\eta)\le t\le\md(\psi_\Delta)$, and $x_t(v)=f\in Z$, then $f(\eta)\in\{0,1\}$ and $f(\eta)=1$ iff
$\KVV(G,p),v\models\eta$.
\end{itemize}
Thus, if $t=\md(\psi_\Delta)$ and $x_t(v)=f$, then $f(\psi_\Delta)$ reveals the truth value of
$\psi_\Delta$ on~$v$. This means that  the computation of $\aA_\Delta$ stops at step $t+1$,
and the output $x_{t+1}(v)$ on node $v$ is $1$ if and only if $\KVV(G,p),v\models\psi_\Delta$. 
In other words, the running time of $\aA_\Delta$ is the constant $\md(\psi_\Delta)+1$, and its output on the input 
$(G,p)$ is the set $\lVert \psi_\Delta\rVert^{\KVV(G,p)}$.
Hence, the sequence $\sA=(\aA_1,\aA_2,\ldots)$ of algorithms solves the graph problem $\Pi_\Psi$,
and we conclude that $\Pi_\Psi\in\VVl$.

\paragraph{Proof of Theorem~\ref{thm:logic}, Part 2.}

We will now consider the proofs of the first halves of (c), (e) and (f). 
In each of these cases, we are given a formula $\psi_\Delta$ in 
the corresponding modal logic, and we define an algorithm $\aA_\Delta$
which simulates the recursive truth definition of $\psi_\Delta$. 
The definitions of the state sets $Y$ and $Z$, as well as the definition of the 
initial state function $z_0$ remain unchanged in all cases. 

However, since the modal operators occurring in subformulas of $\Psi_\Delta$
are different in each of the cases,
the sets $D_j$, $j\in [\Delta]$ have to be redefined accordingly.
Moreover, in cases (e) and (f), we have to remove the markers $j$ from the
messages, since the algorithm $\aA_\Delta$ should be in the class $\Broadcast$.
Thus, the message set $M$ and the message constructing function
$\mu$ have to be redefined for the proof of (e) and (f). 
Finally, in all cases, the clause ($\delta_\Diamond$) in the recursive definition of 
the next state $x_{t+1}(v)$ has to be modified according to the semantics of the
corresponding modal operators.  Below, we list  these modifications for each case
separately.

\begin{itemize}
\item[(c)] For each $j\in [\Delta]$, set $D_j$ is redefined as
    \[
        D_j := \{\eta: \langle (*,j)\rangle_{\ge k}\eta\in\Sigma \text{ for some }k\in\N \}.
    \]
    Clause ($\delta_\Diamond$) is replaced with
    \[
        \tag{$\delta'_\Diamond$}
        g(\eta)=\begin{cases}
            0, \text{ if $f(\vartheta) \ne U$ and $|H|<k$} \\ 
            1, \text{ if $f(\vartheta) \ne U$ and $|H|\ge k$} \\ 
            U, \text{ if $f(\vartheta) = U$}
        \end{cases}
        \text{for }\eta=\langle (*,j)\rangle_{\ge k}\vartheta\in \Sigma,
    \]
    where
    \[
        H = \{i\in[\Delta]: h_i(\vartheta)=(1,j)\}.
    \]
\item[(e)] The definition of $(D_j)_{j\in [\Delta]}$ is replaced with
    \[
        D := \{\eta: \langle (i,*)\rangle\eta\in\Sigma\}.
    \]
    Set $M$ is redefined as
    \[
        M := \{h: h\text{ is a function } D\to \{0,1,U\}\}\cup \{m_0\}.
    \]
    Function $\mu(f,j)$ is redefined to be the restriction of $f$ to the set $D$.
    Clause ($\delta_\Diamond$) is replaced with
    \[
        \tag{$\delta''_\Diamond$}
        g(\eta)=\begin{cases}
            0, \text{ if $f(\vartheta) \ne U$ and $h_i(\vartheta)=0$} \\ 
            1, \text{ if $f(\vartheta) \ne U$ and $h_i(\vartheta)=1$} \\ 
            U, \text{ if $f(\vartheta) = U$}
        \end{cases}
        \text{ for }\eta=\langle (i,*)\rangle\vartheta\in \Sigma.
    \]
    Here we interpret $m_0$ as a function with $m_0(\vartheta) = 0$ for all $\vartheta$.
\item[(f)] The definition of $(D_j)_{j\in [\Delta]}$ is replaced with 
    \[
        D' := \{\eta: \langle (*,*)\rangle_{\ge k}\eta\in\Sigma\text{ for some }k\in\N \}.
    \]
    Set $M$ is redefined as
    \[
        M := \{h: h\text{ is a function } D'\to \{0,1,U\}\}\cup \{m_0\}.
    \]
    Function $\mu(f,j)$ is redefined to be the restriction of $f$ to the set $D'$.
    Clause ($\delta_\Diamond$) is replaced with
    \[
        \tag{$\delta'''_\Diamond$}
        g(\eta)=\begin{cases}
            0, \text{ if $f(\vartheta) \ne U$ and $|H'|<k$} \\ 
            1, \text{ if $f(\vartheta) \ne U$ and $|H'|\ge k$} \\ 
            U, \text{ if $f(\vartheta) = U$}
        \end{cases}
        \text{for }\eta=\langle (*,*)\rangle_{\ge k}\vartheta\in \Sigma,
    \]
    where
    \[
        H' = \{i \in [\Delta] : h_i(\vartheta)=1\}.
    \]
    Again, we interpret $m_0$ as a function with $m_0(\vartheta) = 0$ for all $\vartheta$.
\end{itemize}

It is now straightforward to check that in all the cases $\aA_\Delta$ computes the truth
value of $\psi_\Delta$ correctly in $\md(\psi_\Delta)+1$ steps, whence $\sA = (\aA_1, \aA_2, \dotsc)$
solves $\Pi_\Psi$ in constant time. Furthermore, it is easy to see that in case (c), 
$\aA_\Delta$ is in the class $\Multiset$, whence $\Pi_\psi$ is in $\MVl$. Similarly,
in case (e), $\aA_\Delta$ is in $\Broadcast$, and in case (f) $\aA_\Delta$ is in
$\Multiset\cap\Broadcast$, as desired.

\paragraph{Proof of Theorem~\ref{thm:logic}, Part 3.}

Assume now that $\Pi$ is a graph problem in $\VVl$. Thus, there is a 
sequence $\sA=(\aA_1,\aA_2,\ldots)$
of local algorithms in $\sVector$ such that for every $G\in\F(\Delta)$ and port numbering
$p$ of $G$, the output of $\aA_\Delta$ on $(G,p)$ is in $\Pi(G)$. We will encode
information on the states of computation and messages sent during the execution of $\aA_\Delta$ on an input
$(G,p)$ by suitable formulas of $\MML$. 

Using the definitions of Section~\ref{ssec:distalg}, let $\aA_\Delta=(Y, Z, z_0, M, m_0, \mu, \delta)$, and let $T$ be 
the running time of $\aA_\Delta$. 
We will build a formula
$\psi_\Delta\in \MML(I^{\Delta}_{+,+},\Phi_\Delta)$ simulating $\aA_\Delta$ from the following subformulas:
\begin{itemize}[noitemsep]
\item $\varphi_{z,t}$ for $z\in Y\cup Z$ and $t\in [T]$,
\item $\vartheta_ {m,j,t}$ for $m\in M$, $j\in[\Delta]$ and $t\in [T]$,
\item $\chi_{m,i,j,t}$ for $m\in M$, $i,j\in[\Delta]$ and $t\in [T]$.
\end{itemize}
The intended meaning of these subformulas are given in Table~\ref{tab:logic-mean}, and their recursive definitions are indicated in 
Table~\ref{tab:logic-def}.

\begin{table}
    \centering
    \begin{tabular*}{\columnwidth}{@{\extracolsep{\fill}}ll@{}}
        \toprule
        Subformulas of $\psi_\Delta$ & Algorithm $\aA_\Delta$
        \\
        \midrule
        $\varphi_{z,t}$ is true in world $v$
        & local state $x_t(v)$ is $z$
        \\
        $\vartheta_{m,j,t}$ is true in world $v$
        & node $v$ sends message $m$ to port $j$ in round $t$
        \\
        $\chi_{m,i,j,t}$ is true in world $v$
        & node $v$ receives message $m$ from port $i$ in round $t$,
        \\
        & the message was sent by an adjacent node to port $j$
        \\
        \bottomrule
    \end{tabular*}
    \caption{The intended meaning of the subformulas.}\label{tab:logic-mean}
\end{table}

\begin{table}
    \begin{tabular*}{\columnwidth}{@{\extracolsep{\fill}}l@{ }l@{\hspace{16pt}}l@{}}
        \toprule
        \multicolumn{2}{@{}l}{Recursive definition of the formulas}
        & Execution of $\aA_\Delta$
        \\
        \midrule
        \addlinespace
        $\varphi_{z,0}$: & Boolean combination of $q_i\in\Phi_\Delta$
        & initialisation: \\ && $x_0(u) = z_0(\deg(u))$
        \\ \addlinespace
        $\vartheta_{m,j,t+1}$: & Boolean combination of $\varphi_{z,t}$, $z\in Y\cup Z$
        & local computation: \\ && $m = \mu(x_t(v),j)$
        \\ \addlinespace
        \multicolumn{2}{@{}l}{$\chi_{m,i,j,t+1}: = \langle\alpha\rangle \vartheta_{m,j,t+1}$ with $\alpha=(i,j)$}
        & communication: \\ && construct $\vec{a}_{t+1}(v)$
        \\ \addlinespace
        $\varphi_{z,t+1}$: & Boolean combination of $\varphi_{x,t}$, $x\in Y\cup Z$,
        & local computation:
        \\
        & and $\chi_{m,i,j,t+1}$, $m\in M$, $i,j\in [\Delta]$
        & $x_{t+1}(v) = \delta(x_t(v), \vec{a}_{t+1}(v))$
        \\
        \addlinespace
        \bottomrule
    \end{tabular*}
    \caption{Constructing the formula $\psi_\Delta$, given an algorithm $\aA_\Delta$.}\label{tab:logic-def}
\end{table}

Note that the set $Z$ of intermediate states, as well as the set $M$ of messages,
may be infinite, whence there are potentially infinitely many formulas of the form
$\varphi_{z,t}$, $\vartheta_ {m,j,t}$ and $\chi_{m,i,j,t}$. However, it is easy
to prove by induction on $t$ that there are only finitely many \emph{different} formulas 
in the families
\begin{align*}
    \Psi_t &= \{\varphi_{z,t} : z\in Y\cup Z\}, \\
    \Theta_t &= \{\vartheta_ {m,j,t} : m\in M \text{ and } j\in[\Delta]\}, \\
    \Xi_t &= \{\chi_{m,i,j,t} : m\in M \text{ and } i,j\in[\Delta]\}.
\end{align*}
Indeed, for each $z\in Y\cup Z$, subformula $\varphi_{z,0}$ is a disjunction of the form
$\bigvee_{i\in J}q_i$ for some $J\subseteq [\Delta]$;
here $\bigvee_{i\in\emptyset}q_i$ is understood as some fixed contradictory formula. 
Furthermore, assuming $\Psi_t$ is finite,
there are only finitely many different Boolean combinations of formulas in 
$\Psi_t$, whence $\Theta_{t+1}$ is finite. By the same argument, if $\Theta_{t+1}$ is finite,
then so is $\Xi_{t+1}$, and if $\Psi_t$ and $\Xi_{t+1}$ are finite, then so is 
$\Psi_{t+1}$.

Clearly the formulas $\varphi_{z,t}$, $\vartheta_ {m,j,t}$ and $\chi_{m,i,j,t}$
can be defined in such a way that each of them has its intended meaning.
In particular, given an input $(G,p)$ to the algorithm $\aA_\Delta$, the output on a node $v$ is $1$
if and only if $v\in\lVert\varphi_{1,T}\rVert^{\KVV(G,p)}$. Thus, defining $\psi_\Delta:= \varphi_{1,T}$
we get $\lVert\psi_\Delta\rVert^{\KVV(G,p)}\in\Pi(G)$ for all $G\in\F(\Delta)$ and all port numberings
$p$ of $G$. Hence we conclude that the sequence $\Psi=(\psi_1,\psi_2,\ldots)$ defines a solution to $\Pi$.

As an additional remark, we note that the modal depth of each $\varphi_{z,t}$ is $t$, as an easy 
induction shows. In particular, $\md(\psi_\Delta)$ is equal to the running time $T$ of $\aA_\Delta$. 

\paragraph{Proof of Theorem~\ref{thm:logic}, Part 4.}

To complete the proof, we will now describe the changes needed in the technical details for proving 
the second halves of claims (c), (e) and (f). Thus, assume that $\Pi$ is a graph problem, and
$\sA=(\aA_1,\aA_2,\ldots)$ is a local algorithm which solves $\Pi$ and is in the class $\sMultiset$, $\sBroadcast$ 
or $\sMultiset\cap\sBroadcast$,
respectively. The corresponding modal formula $\psi_\Delta$ is constructed
from subformulas in the same way as in (3) with appropriate modifications in technical details. 

Since algorithms in $\Multiset$ cannot distinguish between the port numbers of incoming 
messages, the subscript $i$ in the formulas $\chi_{m,i,j,t}$ has to be removed in cases of (c) and (f).
On the other hand, the algorithms can count the multiplicities of incoming messages, whence
a new parameter $k\in [\Delta]$ for these formulas is needed. Furthermore, in cases (e) and (f),
the subscript $j$ has to be removed from the formulas $\vartheta_{m,j,t}$ and $\chi_{m,i,j,t}$,
as algorithms in the class $\Broadcast$ cannot send different messages through different ports.
Below, we summarise the modifications in each case separately.

\begin{itemize}
\item[(c)] The formulas $\chi_{m,i,j,t}$ are replaced with $\chi^k_{m,j,t}$, $k\in [\Delta]$. 
The recursive definition of these formulas is as follows:
\[
    \chi^k_{m,j,t+1}:=\langle (*,j)\rangle_{\ge k}\vartheta_{m,j,t+1}.
\]
The formulas $\vartheta_{m,j,t+1}$ and $\varphi_{z,t+1}$ are defined 
as in Table~\ref{tab:logic-def}, with $\chi^k_{m,j,t}$ in place of $\chi_{m,i,j,t}$.
\item[(e)] The formulas $\vartheta_{m,j,t}$ and $\chi_{m,i,j,t}$ are replaced with 
$\vartheta_{m,t}$, and $\chi_{m,i,t}$, respectively. 
The recursive definition of the latter is as follows:
\[
    \chi_{m,i,t+1}:=\langle (i,*)\rangle\vartheta_{m,t+1}.
\]
The formulas $\vartheta_{m,t+1}$ are defined as Boolean combinations of 
$\varphi_{z,t}$, and the formulas $\varphi_{z,t+1}$ are defined as
in Table~\ref{tab:logic-def}.
\item[(f)] The formulas $\vartheta_{m,j,t}$ and $\chi_{m,i,j,t}$ are replaced with 
$\vartheta_{m,t}$, and $\chi^k_{m,t}$, respectively. 
The recursive definition of the latter is as follows:
\[
    \chi^k_{m,t+1}:=\langle (*,*)\rangle_{\ge k}\vartheta_{m,t+1}.
\]
The formulas $\vartheta_{m,t+1}$ are defined as Boolean combinations of 
$\varphi_{z,t}$, and the formulas $\varphi_{z,t+1}$ are defined as
in Table~\ref{tab:logic-def}.
\end{itemize}

As in the proof of claim (b), it is easy to see that in each case the subformulas 
used in the construction of $\psi_\Delta:=\varphi_{1,T}$
can be defined in such a way that they have their intended meanings. 
Thus, for every graph $G\in\F(\Delta)$ and every port numbering $p$ of $G$, the output of $\aA_\Delta$ in $(G,p)$ equals
\[
	\lVert \psi_\Delta\rVert^{\KM_{a,b}(G,p)},
\]
where $a,b \in \{-,+\}$ is selected appropriately for each case. 
Hence, we conclude that the sequence $\Psi=(\psi_1,\psi_2,\ldots)$ defines 
a solution to $\Pi$.

This concludes the proof of Theorem~\ref{thm:logic}. \qed

\bigskip
\noindent
There is a slight asymmetry in the proof of Theorem~\ref{thm:logic}: 
in the first half of the proof the running time of the constructed algorithm $\aA_\Delta$
is $\md(\psi_\Delta)+1$, while in the second half the modal depth of the formula
$\psi_\Delta$ constructed is exactly the running time of the given algorithm $\aA_\Delta$.
However, this mismatch can be rectified by modifying the proof of the first part.
We did not write this modified proof simply to avoid unnecessary technicalities.

Theorem~\ref{thm:logic} gives us a tool for proving that a given graph
problem $\Pi$ is not in one of the classes considered in this paper. The idea
is to use bisimulation for showing that the corresponding modal logic cannot
define a solution for $\Pi$. At first it may appear that this tool can be applied
only for the constant-time versions of the classes, as the logical
characterisations in Theorem~\ref{thm:logic} are valid only in the constant-time
case. However, in the following corollary we show that the method based on 
bisimulation  can be used also in the general case. 
In principle, this result is valid for all seven classes, but we formulate it here only
for $\VV$, $\VB$ and $\SB$; these are the cases we use later in Section~\ref{sec:separ}. 
We remind the reader that throughout this section we focus on the case of binary outputs, i.e., $Y = \{0,1\}$, in which case we can interpret a solution $S \in \Pi$ as a subset $S \subseteq V$.

\begin{corollary}\label{cor:bisim}
Let $G=(V,E)\in\F(\Delta)$ be a graph, $X\subseteq V$, and let $\Pi$ be a 
graph problem such that  
for every $S\in\Pi(G)$, there are $u,v\in X$ with $u\in S$ and $v\notin S$.
\begin{enumerate}
    \item If there is a port numbering $p$ of $G$ such that all nodes in $X$ are bisimilar in the model $\KVV(G,p)$, then $\Pi$ is not in the class $\VV$.
    \item If there is a port numbering $p$ of $G$ such that all nodes in $X$ are bisimilar in the model $\KVB(G,p)$, then $\Pi$ is not in the class $\VB$.
    \item If there is a port numbering $p$ of $G$ such that all nodes in $X$ are bisimilar in the model $\KMB(G,p)$, then $\Pi$ is not in the class $\SB$.
\end{enumerate}
\end{corollary}
\begin{proof}
We prove only claim (a); the other claims can be proved in the same way. Let $\sA = (\aA_1, \aA_2, \dotsc)$ be any algorithm in $\sVector$, and let $\Delta$ be the maximum degree of $G$.

The key observation is that there is a \emph{local} algorithm $\aB_\Delta \in \Vector$ such that $\aB_\Delta$ and $\aA_\Delta$ produce the same output $S$ in $(G,p)$. We can obtain such a local algorithm $\aB_\Delta$ from algorithm $\aA_\Delta$ by adding a counter that stops the computation after $T$ steps, where $T$ is the running time of $\aA_\Delta$ on $(G,p)$.

As we have a local algorithm $\aB_\Delta$ that produces output $S$ in $(G,p)$, by Theorem~\ref{thm:logic}b there is also a formula $\psi\in\MML(I^\Delta_{+,+})$ such that
\[
	S=\lVert\psi \rVert^{\KVV(G,p)}.
\]
By assumption, all nodes in $X$ are bisimilar in the model $\KVV(G,p)$. By Fact \ref{fact:bisim}, there can be no $u,v\in X$ such that $u\in S$ and $v\notin S$. Hence we have $S \notin \Pi(G)$, and we conclude that $\sA$ cannot solve $\Pi$.
\end{proof}

\section{Relations between the Classes}

Now we are ready to prove relations \eqref{eq:order} and \eqref{eq:orderl} that we gave in Section~\ref{sec:contrib}.

\subsection{Equality \texorpdfstring{$\MV = \SV$}{MV = SV}}

Theorem~\ref{thm:multisetset} is the most important technical contribution of this work. Informally, it shows that \emph{outgoing} port numbers necessarily break symmetry even if we do not have \emph{incoming} port numbers---provided that we are not too short-sighted.

\begin{theorem}\label{thm:multisetset}
    Let $\Pi$ be a graph problem and let $T\colon \N \times \N \to \N$. Assume that there is an algorithm $\sA \in \sMultiset$ that solves $\Pi$ in time $T$. Then there is an algorithm $\sB\in \sSet$ that solves $\Pi$ in time $T + O(\Delta)$.
\end{theorem}

To prove Theorem~\ref{thm:multisetset}, we define the following local algorithm $\aC_\Delta \in \Set$. Each node $v$ constructs two sequences, $\beta_t(v)$ and $B_t(v)$ for $t = 0, 1, \dotsc, 2\Delta$. Before the first round, each node $v$ sets $\beta_0(v) = \emptyset$ and $B_0(v) = \emptyset$. Then in round $t = 1, 2, \dotsc, 2\Delta$, each node $v$ does the following:
\begin{enumerate}[noitemsep,label=(\arabic*)]
    \item Set $\beta_t(v) = (\beta_{t-1}(v), B_{t-1}(v))$.
    \item For each port $i$, send $(\beta_t(v), \deg(v), i)$ to port $i$.
    \item Let $B_t(v)$ be the set of all messages received by $v$.
\end{enumerate}

Let $G = (V,E) \in \F(\Delta)$, and let $p$ be a port numbering of graph $G$. We will analyse the execution of $\aC_\Delta$ on $(G,p)$. If $p((v,i)) = (u,j)$, we define that $\pi(v,u) = i$. That is, $\pi(v,u)$ is the outgoing port number in $v$ that is connected to $u$. Let
\[
    m_t(u,v) = \bigl(\beta_t(u),\, \deg(u),\, \pi(u,v)\bigr)
\]    
denote the message that node $u$ sends to node $v$ in round $t$; it follows that $m_t(u,v) \in B_t(v)$ for all $\{u,v\} \in E$.

We begin with a following technical lemma. To pinpoint the key notion, let us
call $u$ and $w$ a \emph{pair of indistinguishable neighbours} of~$v$ in round~$t$, if
they are distinct neighbours of~$v$ such that  
\[
    \beta_t(u) = \beta_t(w)\text{, }
    \deg(u) = \deg(w)\text{, and }
    \pi(u,v) = \pi(w,v). 
\]
This is the same as saying that the node $v$ receives the same message from $u$ and $w$ in round $t$.
Let us say that $u$ and $w$ are a \emph{pair of indistinguishable neighbours of order~$k$}
if further it holds that $v$ has $k$ distinct neighbours $v_1, v_2, \dotsc, v_k$  such that
   \[
\beta_t(u) = \beta_t(w) = \beta_t(v_i)\text{ for all }i = 1, 2, \dotsc, k.
   \]
Here, $u$ or $w$ may belong to the set~$\{v_1, v_2, \dotsc, v_k\}$.  Note, however, that we do not
require each pair $(v_i,v_j)$ to be a pair of indistinguishables.

\begin{lemma}\label{lem:multisetset1}
Suppose that $u$ and $w$ are a pair of indistinguishable neighbours of~$v$ of order $k$ in round~$t\ge 4$. 
Then $u$ and $w$ are a pair of indistinguishable neighbours of~$v$ of order~$k+1$ in round~$t-2$.
\end{lemma}
\begin{proof}
    From $\beta_t(u) = \beta_t(w)$ it follows that $\beta_{t-2}(u) = \beta_{t-2}(w)$. This implies $m_{t-2}(u,v) = m_{t-2}(w,v)$.

    For all $i=1,2,\dotsc,k$, node $v_i$ receives the message
    \[
        m_{t-1}(v,v_i) = \bigl(\beta_{t-1}(v),\, \deg(v),\, \pi(v,v_i)\bigr)
    \]
    from $v$ in round $t-1$. By assumption, we have $\beta_t(v_i) = \beta_t(v_j)$ for all $i$ and~$j$, which implies $B_{t-1}(v_i) = B_{t-1}(v_j)$. Now $m_{t-1}(v,v_i) \in B_{t-1}(v_i)$ implies $m_{t-1}(v,v_j) \in B_{t-1}(v_i)$ for all $i$ and~$j$.

    In any port numbering, we have $\pi(v,v_i) \ne \pi(v,v_j)$ for $i \ne j$. Therefore $m_{t-1}(v,v_i) \ne m_{t-1}(v,v_j)$, and $B_{t-1}(v_1)$ contains $k$ distinct messages. That is, node $v_1$ has $k$ distinct neighbours, $u_1, u_2, \dotsc, u_k$, such that
    \[
    \begin{split}
        \bigl(\beta_{t-1}(u_i),\, \deg(u_i),\, \pi(u_i,v_1)\bigr)
        &= m_{t-1}(u_i,v_1)
         = m_{t-1}(v,v_i) \\
        &= \bigl(\beta_{t-1}(v),\, \deg(v),\, \pi(v,v_i)\bigr).
    \end{split}
    \]
    In particular, $\beta_{t-1}(u_i) = \beta_{t-1}(v)$ for all $i$.
    
    Now let us investigate the messages that the nodes $u_i$ receive in round $t-2$. We have
    \[
        m_{t-2}(v_1,u_i) = \bigl(\beta_{t-2}(v_1),\, \deg(v_1),\, \pi(v_1,u_i)\bigr).
    \]
    However, $\beta_{t-1}(u_i) = \beta_{t-1}(v)$ implies $B_{t-2}(u_i) = B_{t-2}(v)$ for all $i$. In particular, $m_{t-2}(v_1,u_i) \in B_{t-2}(v)$ for all $i$. Now $\pi(v_1,u_i) \ne \pi(v_1,u_j)$ implies $m_{t-2}(v_1,u_i) \ne m_{t-2}(v_1,u_j)$ for all $i \ne j$.
    
    To summarise, $v$ receives the following messages in round $t-2$:
    $k$ distinct messages,
    \[
        m_{t-2}(v_1,u_i) = \bigl(\beta_{t-2}(v_1),\, \deg(v_1),\, \pi(v_1,u_i)\bigr) \text{ for } i = 1, 2, \dotsc, k,
    \]
    and two identical messages,
    \[
        m_{t-2}(u,v) = m_{t-2}(w,v) = \bigl(\beta_{t-2}(u),\, \deg(u),\, \pi(u,v)\bigr).
    \]
    Moreover, $\beta_{t-2}(v_1) = \beta_{t-2}(u)$. Hence $v$ receives at least $k+1$ messages in round $t-2$, each of the form $(\beta_{t-2}(u), \cdot, \cdot)$. Therefore $v$ has at least $k+1$ distinct neighbours $v'_i$ with $\beta_{t-2}(u) = \beta_{t-2}(v'_i)$.
\end{proof}

\begin{lemma}\label{lem:multisetset2}
If a node $v$ has two indistinguishable neighbours $u$ and $w$ of order~$k$ in round $2t$, 
then $v$ has at least $t+k-1$ neighbours.  Consequently, no node has a pair of indistinguishable
neighbours in round~$2\Delta$.
\end{lemma}

\begin{proof}
The proof of the first claim is by induction on $t$. The base case $t=1$ is trivial. For the inductive step, apply Lemma~\ref{lem:multisetset1}.
For the consequence, we observe that if $u$ and $w$ were a pair of indistinguishable neighbours of~$v$,
then the first claim would imply that $\deg(v)\ge \Delta+1$, which is a contradiction,
as the maximum degree of~$G$ is at most $\Delta$.
\end{proof}

To summarise: 
$m_{2\Delta}(u,v) \ne m_{2\Delta}(w,v)$ whenever $u$ and $w$ are two distinct neighbours of $v$ in~$G$. 
In particular,
\[
    \bigl(\beta_{2\Delta}(u),\,\deg(u),\,\pi(u,v)\bigr) \ne
    \bigl(\beta_{2\Delta}(w),\,\deg(w),\,\pi(w,v)\bigr).
\]

Once we have finished running $\aC_\Delta$, which takes $O(\Delta)$ time, we can \emph{simulate} the execution of $\aA_\Delta \in \Multiset$ with an algorithm $\aB_\Delta \in \Set$ as follows: if a node~$u$ in the execution of $\aA_\Delta$ sends the message $a$ to port $i$, algorithm $\aB_\Delta$ sends the message
\[
    \bigl(\beta_{2\Delta}(u),\deg(u),\,i,\,a\bigr)
\]
to port $i$. Now all messages received by a node are distinct. Hence given the set of messages received by a node $v$ in $\aB_\Delta$, we can reconstruct the multiset of messages received by $v$ in $\aA_\Delta$. This concludes the proof of Theorem~\ref{thm:multisetset}.

\begin{corollary}\label{cor:multisetset}
    We have $\MV = \SV$ and $\MVl = \SVl$.
\end{corollary}
\begin{proof}
    Immediate from Theorem~\ref{thm:multisetset}.
\end{proof}

\begin{remark}
With minor modifications, the proof of Theorem~\ref{thm:multisetset} would also imply $\VV = \MV = \SV$. However, as we will see next, there is also a more direct way to prove $\VV = \MV$. The proof in the following section avoids the additive overhead in running time (but the overhead in message size may be much larger).
\end{remark}

\subsection{Equalities \texorpdfstring{$\VB = \MB$}{VB = MB} and \texorpdfstring{$\VV = \MV$}{VV = MV}}

The following theorem is implicit in prior work~\cite[Section 5]{astrand10vc-sc}; we give a bit more detailed version for the general case here. The basic idea is that algorithm $\sB$ augments each message with the full communication history, and orders the incoming messages lexicographically by the communication histories---the end result is equal to the execution of algorithm $\sA$ in the same graph $G$ for a very specific choice of incoming port numbers.

\begin{theorem}\label{thm:vectormultiset}
    Let $\Pi$ be a graph problem and let $T\colon \N \times \N \to \N$. Assume that there is an algorithm $\sA \in \sVector$ that solves $\Pi$ in time $T$. Then there is an algorithm $\sB \in \sMultiset$ that solves $\Pi$ in time $T$.
\end{theorem}

\begin{proof}
Let $\sA = (\aA_1, \aA_2, \dotsc) \in \sVector$ be an algorithm, and let $G \in \F(\Delta)$ be a graph of maximum degree at most $\Delta$. Consider a port numbering $p$ of $G$, and the execution of $\aA_\Delta$ on $(G,p)$. For each port $(u,i) \in P(G)$, let
\[
    \beta_t(u,i) = (a_1(u,i), a_2(u,i), \dotsc, a_t(u,i))
\]
be the full history of messages that node $u$ received from port $i$ in rounds $1, 2, \dotsc, t$. Let $<$ be the lexicographical order of such vectors, that is, $\beta_t(u,i) < \beta_t(v,j)$ if there is a time $\ell$ such that 
$a_\ell(u,i) <_M a_\ell(v,j)$ and $a_k(u,i) = a_k(v,j)$ for $k = 1, 2, \dotsc, \ell - 1$. Here $<_M$ is a fixed order
of the message set $M$ of $\aA_\Delta$.

We say that $p$ is \emph{compatible with the message history up to time $t$} if $\beta_t(u,i) \le \beta_t(u,j)$ for all nodes $u \in V$ and all $i < j$. Clearly, if $p$ is compatible with the message history up to time $t$, it is also compatible with the message history up to time $t-1$.

Now fix any port numbering $p_0$ of $G$. Let $\Pp_0$ be the family of all port numberings $p$ of $G$ such that for each port $(u,i) \in P(G)$ there are $v$, $j$, and $k$ such that $p((u,i)) = (v,j)$ and $p_0((u,i)) = (v,k)$. Put otherwise, any $p \in \Pp_0$ is equivalent to $p_0$ from the perspective of a $\Multiset$ algorithm. We make the following observations:
\begin{enumerate}[itemsep=0ex]
    \item $\Pp_0$ is non-empty, and each $p \in \Pp_0$ is compatible with the message history up to time $0$.
    \item State vector $x_0$ at time $0$ does not depend on the choice of $p \in \Pp_0$.
    \item The message sent by a node $v$ to port $j$ in round $1$ does not depend on the choice of $p \in \Pp_0$.
    \item There is at least one $p \in \Pp_0$ that is compatible with the message history up to time $1$.
\end{enumerate}
Now let $\Pp_t \subseteq \Pp_{t-1}$ consist of \emph{all} port numberings $p \in \Pp_{t-1}$ that are compatible with the message history up to time $t$. By induction, we have:
\begin{enumerate}[itemsep=0ex]
    \item $\Pp_t$ is non-empty, and each $p \in \Pp_t$ is compatible with the message history up to time $t$.
    \item The vector $\vec{a}_t(u)$ of messages received by $u$ in round $t$ does not depend on the choice of $p \in \Pp_t$. State vector $x_t$ at time $t$ does not depend on the choice of $p \in \Pp_t$.
    \item The message sent by a node $v$ to port $j$ in round $t+1$ does not depend on the choice of $p \in \Pp_t$.
    \item There is at least one $p \in \Pp_t$ that is compatible with the message history up to time $t+1$.
\end{enumerate}

Let $T = T(\Delta, |V|)$. In particular, $\aA_\Delta$ stops in time $T$ in $(G,p)$ for any $p \in \Pp_T$. Intuitively, a port numbering $p \in \Pp_T$ is constructed as follows: we have copied the \emph{outgoing} port numbers from a given port numbering $p_0$, and we have adjusted the \emph{incoming} port numbers so that $p$ becomes compatible with the message history up to time $T$. This choice of incoming port numbers is particularly convenient from the perspective of the $\Multiset$ model:
$\beta_t(u,i) < \beta_t(u,j)$ implies $i < j$, and
$\beta_t(u,i) = \beta_t(u,j)$ implies $a_t(u,i) = a_t(u,j)$.
That is, if we know the multiset
\[
    \multiset( \beta_t(u,1), \beta_t(u,2), \dotsc, \beta_t(u,\Delta) ),
\]
we can reconstruct the vector~$\vec{a}_t(u)$.

We design an algorithm $\aB_\Delta \in \Multiset$ with the following property: the execution of $\aB_\Delta$ on $(G,p_0)$ simulates the execution of $\aA_\Delta$ on $(G,p)$, where $p \in \Pp_T$. Note that the output of $\aA_\Delta$ does not depend on the choice of $p \in \Pp_T$. As the output of $\aA_\Delta$ is in $\Pi(G)$ for any port numbering of $G$, it follows that the output of $\aB_\Delta$ is also in $\Pi(G)$.

The simulation works as follows. For each port $(v,j) \in P(G)$, algorithm $\aB_\Delta$ keeps track of all messages sent by node $v$ to port $j$ in $\aA_\Delta$. Each outgoing message is augmented with the full message history. Hence in each round $t$, a node $u$ can reconstruct the unique vector $\vec{a}_t(u)$ that matches the execution of $\aA_\Delta$ on $(G,p)$ for any $p \in \Pp_t$.
\end{proof}

\begin{theorem}\label{thm:vectormultisetb}
    Let $\Pi$ be a graph problem and let $T\colon \N \times \N \to \N$. Assume that there is an algorithm $\sA \in \sBroadcast$ that solves $\Pi$ in time $T$. Then there is an algorithm $\sB \in \sMultiset \cap \sBroadcast$ that solves $\Pi$ in time $T$.
\end{theorem}

\begin{proof}
This is similar to the proof of Theorem~\ref{thm:vectormultiset}. In the simulation, each node keeps track of the history of broadcasts.
\end{proof}

\begin{corollary}\label{cor:vectormultiset}
    We have $\VB = \MB$, $\VBl = \MBl$, $\VV = \MV$, and $\VVl = \MVl$.
\end{corollary}
\begin{proof}
    Follows from Theorems \ref{thm:vectormultiset} and \ref{thm:vectormultisetb}.
\end{proof}

\begin{remark}
Boldi et al.~\cite{boldi96symmetry} and Yamashita and Kameda~\cite{yamashita99leader} already give simulation results that, in essence, imply $\VB = \MB$ and $\VV = \MV$ (albeit for a slightly different model of computation). However, in prior work, the simulation overhead is linear in the number of nodes; in particular, it does not imply $\VBl = \MBl$ or $\VVl = \MVl$.
\end{remark}

\subsection{Separating the Classes}\label{sec:separ}

Trivially, $\SB \subseteq \MB \subseteq \MV$ and $\SBl \subseteq \MBl \subseteq \MVl$. Together with Corollaries \ref{cor:multisetset} and \ref{cor:vectormultiset} these imply
\begin{gather*}
    \SB \subseteq \MB = \VB \subseteq \SV = \MV = \VV, \\
    \SBl \subseteq \MBl = \VBl \subseteq \SVl = \MVl = \VVl.
\end{gather*}
Now we proceed to show that the subset relations are proper. We only need to come up with a graph problem that separates a pair of classes---here the connections to modal logic and bisimulation are a particularly helpful tool. Many of the separation results are already known by prior work (in particular, see Yamashita and Kameda~\cite{yamashita99leader}), but we give the proofs here to demonstrate the use of bisimulation arguments.

For the case of $\VB \ne \SV$, the separation is easy: we can consider the problem of breaking symmetry in a star graph.

\begin{theorem}\label{thm:starleaf}
    There is a graph problem $\Pi$ such that $\Pi \in \SVl$ and $\Pi \notin \VB$.
\end{theorem}
\begin{proof}
An appropriate choice of $\Pi$ is the (artificial) problem of selecting a leaf node in a star graph. More formally, we have the set of outputs $Y = \{0,1\}$. We define $\Pi(G)$ as follows, depending on $G$:
\begin{enumerate}
    \item $G = (V,E)$ is a $k$-star for a $k > 1$. That is, $V = \{ c, v_1, v_2, \dotsc, v_k \}$ and $E = \{ \{ c, v_i \} : i = 1, 2, \dotsc, k \}$. Then we have $S \in \Pi(G)$ if $S\colon V \to Y$, $S(c) = 0$, and there is a $j$ such that $S(v_j) = 1$ and $S(v_i) = 0$ for all $i \ne j$.
    \item $G = (V,E)$ is not a star. Then we do not restrict the output, i.e., $S \in \Pi(G)$ for any function $S\colon V \to Y$.
\end{enumerate}

It is easy to design a local algorithm $\sA \in \sSet$ that solves $\Pi$: First, all nodes send message $i$ to port $i$ for each $i$; then a node outputs $1$ if it has degree $1$ and if it received the set of messages $\{1\}$. Thus, $\Pi$ is in $\SVl$.

We use Corollary~\ref{cor:bisim}b for proving that $\Pi$ is not in $\VB$. Let $G=(V,E)$ be a $k$-star, 
and let $X\subseteq V$ be the set of leaf nodes of $G$. Then $\Pi$ and $X$ satisfy the assumption
in Corollary~\ref{cor:bisim}. Furthermore, it is easy to see that given any port numbering $p$ of $G$,
all nodes in $X$ are bisimilar in the model $\KVB(G,p)$.
\end{proof}

\begin{corollary}\label{cor:starleaf}
    We have $\VB \ne \SV$ and $\VBl \ne \SVl$.
\end{corollary}
\begin{proof}
    Follows from Theorem~\ref{thm:starleaf}.
\end{proof}

To show that $\SB \ne \MB$, we can consider, for example, the problem of identifying nodes that have an odd number of neighbours with odd degrees.

\begin{theorem}\label{thm:sumdeg}
    There is a graph problem $\Pi$ such that $\Pi \in \MBl$ and $\Pi \notin \SB$.
\end{theorem}
\begin{proof}
We define $\Pi$ as follows. Let $G = (V,E)$ and $S\colon V \to \{0,1\}$. We have $S \in \Pi(G)$ if the following holds: $S(v) = 1$ iff $v$ is a node with an odd number of neighbours of an odd degree.

The problem is trivially in $\MBl$: first each node broadcasts the parity of its degree, and then a node outputs $1$ if it received an odd number of messages that indicate the odd parity.

To see that the problem is not in $\SB$, it is sufficient to argue that the white nodes in the following graphs are bisimilar, yet they are supposed to produce different outputs.
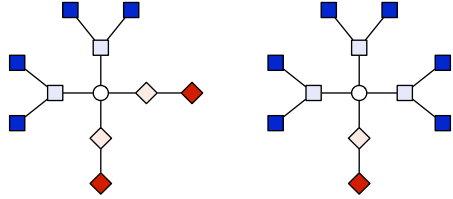
\begin{center}
    \includegraphics[page=\PSumDeg]{figs.pdf}
\end{center}
More precisely, we can partition the nodes in five equivalence classes (indicated with the shading and shapes in the above illustration), and the nodes in the same equivalence class are bisimilar in the Kripke model $\KMB(G,p)$; recall that the model is independent of the choice of the port numbering~$p$. 
Thus, we can apply Corollary~\ref{cor:bisim}c with $X$ consisting of the two white nodes.
\end{proof}

\begin{corollary}\label{cor:sumdeg}
    We have $\SB \ne \MB$ and $\SBl \ne \MBl$.
\end{corollary}
\begin{proof}
    Follows from Theorem~\ref{thm:sumdeg}.
\end{proof}

Finally, to separate $\VV$ and $\VVc$, we make use of the fact that there are graphs
$G$ such that some inconsistent port numbering of $G$ is totally symmetric, while
all consistent port numberings of $G$ necessarily break symmetry between nodes.

We start by proving that any regular graph has a totally symmetric port numbering.
Recall that a graph $G$ is \emph{$k$-regular} if $\deg(v)=k$ for every node $v$ of $G$. Furthermore,
$G$ is \emph{regular} if it is $k$-regular for some $k\in\N$.
Recall also that a  \emph{$1$-factor} (or \emph{perfect matching}) of a
graph $G=(V,E)$ is a set $F\subseteq E$ of edges such that every node 
$v\in V$ has degree $1$ in the graph $(V,F)$.

\begin{lemma}\label{lem:regbisim}
If $G$ is a regular graph, then there is a port numbering $p$ of $G$ such that
all nodes of $G$ are bisimilar in the model $\KVV(G,p)$.  
\end{lemma}
\begin{proof}
Assume that $G=(V,E)$ is $k$-regular. Let $V_s=V\times \{s\}$ for $s\in \{1,2\}$,
and let $E^*=\{\{(u,1),(v,2)\}: \{u,v\}\in E\}$. Then $G^*=(V_1\cup V_2, E^*)$
is a $k$-regular bipartite graph; see Figure~\ref{fig:regbisim}. It is a well-known corollary of
Hall's marriage theorem~\cite[Section 2.1]{diestel10graph} that the edge set of
any such graph is the union of $k$ mutually disjoint $1$-factors. Thus, there are 
sets $E_1,\ldots, E_k\subseteq E^*$ such that $E_i\cap E_j=\emptyset$ whenever $i\ne j$,
and each $E_i$ is a one-to-one correspondence between the sets $V_1$ and $V_2$.

\begin{figure}
    \centering
    \includegraphics[page=\PRegularVv]{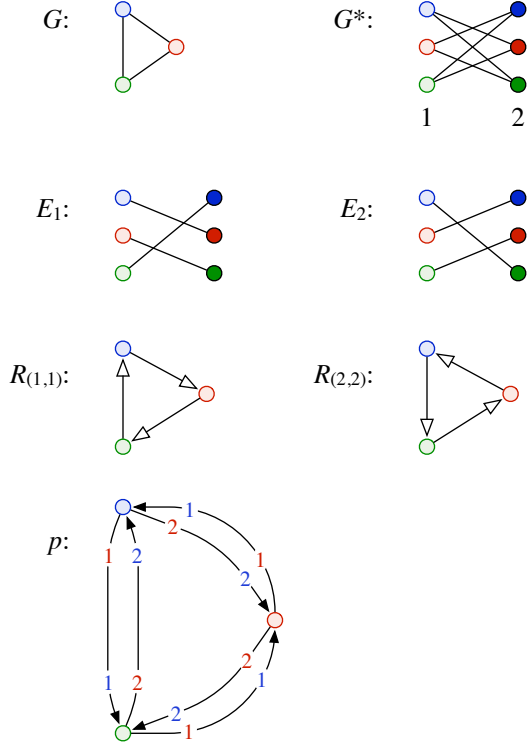}
    \caption{An illustration of the proof of Lemma~\ref{lem:regbisim}.}\label{fig:regbisim}
\end{figure}

Instead of defining a port numbering $p$, we use the sets $E_i$ to define a Kripke model
\[
    \KM=(V, (R_\alpha)_{\alpha\in I^k_{+,+}}, \tau).
\]
For each $i\in [k]$, we let 
$R_{(i,i)}=\{(u,v): \{(u,1),(v,2)\}\in E_i\}$, and if $i\ne j$, we let $R_{(i,j)}=\emptyset$.
Furthermore, we let $\tau$ to be as in the definition of  the models $\KM_{a,b}(G,p)$.
Clearly, there is a port numbering $p$ such that $\KM=\KVV(G,p)$.
Moreover, for every node $u\in V$ the set $\{v\in V: (u,v)\in R_{(i,j)}\}$ is nonempty
if and only if $i=j$. Using this, it is easy to see that the full relation $Z=V\times V$ 
is a bisimulation, whence all nodes are bisimilar in the model $\KVV(G,p)$.
\end{proof}

Note that the converse of Lemma~\ref{lem:regbisim} is true as well: if 
$u$ and $v$ are nodes in $G$ such that $\deg(u)\ne \deg(v)$, then $u$ and $v$
obviously cannot be bisimilar in the model $\KVV(G,p)$ for any port numbering $p$.

Our next aim is to show that there is a class of regular graphs $G$ such that all consistent port numberings of $G$ break symmetry---this is known from prior work~\cite{yamashita99leader} but we give a proof here for completeness.

\begin{lemma}\label{lem:reg-1-factor}
If $G$ is a $k$-regular graph for an odd $k$, and there is a consistent port numbering 
$p$ of $G$ such that all its nodes are bisimilar in the model $\KVV(G,p)$, then $G$
has a $1$-factor.  
\end{lemma}
\begin{proof}
Let $p$ be a consistent port numbering of $G=(V,E)$ such that all nodes of $G$ are bisimilar
in $\KVV(G,p)$. Let  $F\subseteq [k]^2$ be the relation  $\{(i,j)\in [k]^2: R_{(i,j)}\ne \emptyset\}$.
Then $F$ is a function, since otherwise there are $u, v, u', v'\in V$ and such that
$(u,v)\in R_{(i,j)}$ and $(u',v')\in R_{(i,j')}$ for some $i,j,j'\in [k]$ with $j\ne j'$, 
which would imply that  $u$ and $u'$ are not bisimilar. Note further, that  by consistency of $p$,
relation $F$ is symmetric: if $(i,j)\in F$, then $(j,i)\in F$. Thus, $F$ is a permutation of $[k]$ such that
$F^{-1}=F$. Since $k$ is odd, there exists $i\in [k]$ such that $(i,i)\in F$. It is now easy
to see that the relation $R_{(i,i)}$ is a $1$-factor of $G$.
\end{proof}

By the previous two lemmas, each regular graph of odd degree and without $1$-factors
has an inconsistent symmetric port numbering, but no consistent symmetric port numberings.
In the proof of the separation result, we also need the assumption that all graphs we consider
are connected. Thus, we define
$\mathcal G$ to be the class of all connected regular graphs of odd degree which 
do not have a $1$-factor. It is easy to construct $k$-regular graphs in ${\mathcal G}$
for each odd degree $k\ge 3$. The graph illustrated in Figure~\ref{fig:vv-vvc}a is an
example with $k=3$. Figure~\ref{fig:vv-vvc}b shows an example of an inconsistent 
symmetric port numbering of the same graph.

\begin{figure}
    \centering
    \includegraphics[page=\PVvVvc]{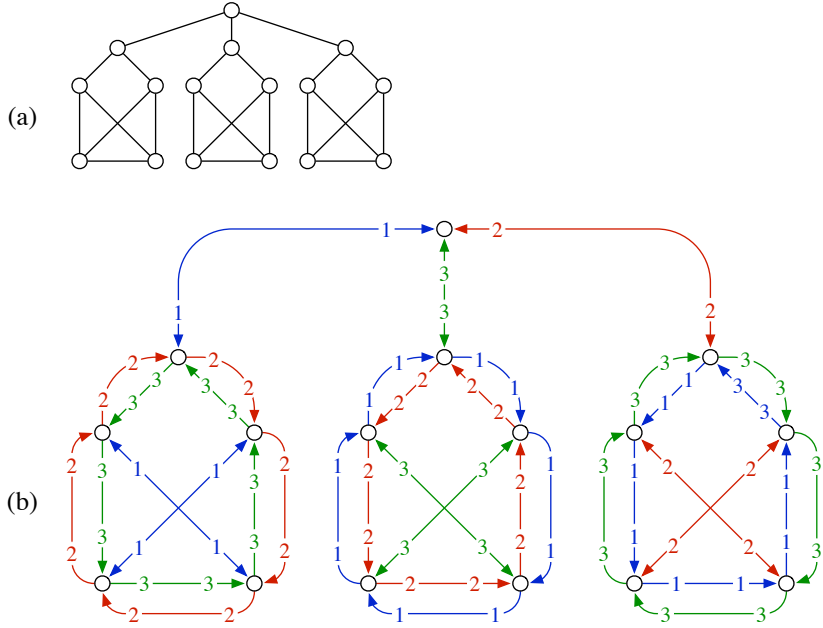}
    \caption{(a)~A 3-regular graph $G$ that does not have a 1-factor \cite[Figure~5.10]{bondy76graph-theory}. (b)~A symmetric port numbering of $G$.}\label{fig:vv-vvc}
\end{figure}

\begin{theorem}\label{thm:vv-vvc}
    There is a graph problem $\Pi$ such that $\Pi \in \VVcl$ and $\Pi \notin \VV$.
\end{theorem}
\begin{proof}
We define the graph problem $\Pi$ as follows: For all graphs $G=(V,E)\in{\mathcal G}$, 
$\Pi(G)$ consists of all non-constant functions $S\colon V \to \{0,1\}$, that is, we have $u,v \in V$ 
with $S(u) \ne S(v)$. For all graphs $G\notin{\mathcal G}$, $\Pi(G)$ consists of all functions 
$S\colon V \to \{0,1\}$.

Let us first prove that the problem is in $\VVcl$. Let $G=(V,E)\in\F(\Delta)$ be a graph
and $p$ a consistent port numbering of $G$. We define the local type of a node $v\in V$ to
be the tuple $t(v)=(j_1,j_2,\dotsc,j_\Delta)$, where $j_i$ is the number of the port of a 
neighbour $u$ to which the port $(v,i)$ is connected if $i\in [\deg(v)]$, and $j_i=0$
for $i>\deg(v)$. Fix some linear ordering $\le$ of the local types. Then there is a local
algorithm $\aA_\Delta\in\Vector$ such that its output on a node $v\in V$ is $1$ if
$t(u)\le t(v)$ for all neighbours $u$ of $v$, and $0$ otherwise: $\aA_\Delta$
computes first in one step the local types of the nodes, and then in a second step 
it sends the types to neighbouring nodes, compares the types, and outputs either 
$0$ or $1$ depending on the comparison.

The crucial observation now is that if $G\in {\mathcal G}$, then $G$ has nodes
with different local types. This is seen as follows.
If the local types of all nodes of $G$ are the same, then it is easy to see that 
all nodes are bisimilar in the model $\KVV(G,p)$. Thus, by Lemma~\ref{lem:reg-1-factor}, 
either $G$ is not $k$-regular for any odd $k$, or $G$ has a $1$-factor. 

Assume then that $G\in {\mathcal G}$ and $p$ is a consistent port numbering of $G$, 
and consider the output $S\colon V\to \{0,1\}$ that is produced by $\aA_\Delta$ in $(G,p)$. Since the local types
of all nodes are not the same and $G$ is connected, there are nodes $u,v\in V$
such that $t(u)<t(v)$, and $t(v)$ is maximal w.r.t.\ the ordering $\le$. This means
that $S(u)=0$ and $S(v)=1$, whence $S\in \Pi(G)$. We conclude that
the sequence $\sA=(\aA_1,\aA_2,\ldots)$ of algorithms solves $\Pi$ assuming consistency.

To see that $\Pi$ is not in $\VV$,
consider a graph $G=(V,E)\in {\mathcal G}$. By Lemma~\ref{lem:regbisim}, there exists a port numbering
$p$ of $G$ such that all nodes of $G$ are bisimilar in the model $\KVV(G,p)$ (as seen above, 
$p$ is inconsistent). The claim follows now from Corollary~\ref{cor:bisim}a, since
the graph problem $\Pi$ and the set $X=V$ clearly satisfy its assumption.
\end{proof}

\begin{corollary}\label{cor:vv-vvc}
    We have $\VVc \ne \VV$ and $\VVcl \ne \VVl$.
\end{corollary}
\begin{proof}
    Follows from Theorem~\ref{thm:vv-vvc}.
\end{proof}

\subsection{Conclusions}

In summary, we have established that the classes we have studied form a linear order of length four:
\begin{gather}
    \tag{\ref{eq:order}}
    \SB \subsetneq \MB = \VB \subsetneq \SV = \MV = \VV \subsetneq \VVc. \\
    \tag{\ref{eq:orderl}}
    \SBl \subsetneq \MBl = \VBl \subsetneq \SVl = \MVl = \VVl \subsetneq \VVcl.
\end{gather}
As a corollary of \eqref{eq:orderl} and Theorem~\ref{thm:logic}, we can make, for example, the following observations.
\begin{enumerate}
    \item $\MML$ captures the same class of problems on $\CVV$ and $\CMV$.
    \item Both $\MML$ and $\GMML$ capture the same class of problems on $\CMV$.
    \item The class of problems captured by $\MML$ becomes strictly smaller if we replace $\CMV$ with $\CVB$.
    \item $\MML$ on $\CVB$ captures the same class of problems as $\GML$ on $\CMB$.
\end{enumerate}

\paragraph{Open Questions Related to Equalities.}

In the proofs of Corollaries \ref{cor:multisetset} and \ref{cor:vectormultiset}, our main focus was on devising a simulation scheme in which the simulation overhead is only proportional to maximum degree $\Delta$ and running time $T$---this implies that local algorithms of a stronger model can be simulated with local algorithms in a weaker model. However, in our approach the simulation overhead is large in terms of message size. It is an open question if such a high overhead is necessary.

\paragraph{Open Questions Related to Separations.}

To keep the proofs of Theorems \ref{thm:starleaf}, \ref{thm:sumdeg}, and \ref{thm:vv-vvc} as simple as possible, we introduced graph problems that were highly contrived. An interesting challenge is to come up with \emph{natural} graph problems that could be used to prove the same separation results. It should be noted that prior work \cite{boldi96symmetry, yamashita99leader} presents some separation results that use a natural graph problem---leader election. However, leader election is a global problem; it cannot be solved in $\VVcl$, and hence we cannot use it to separate any of the constant-time versions of the classes.

Another challenge is to come up with \emph{decision problems} that separate the classes. Indeed, it is not known if the separation results hold if we restrict ourselves to decision problems.

\section*{Acknowledgements}

This work is an extended and revised version of a preliminary conference report~\cite{hella12weak-models}. We thank anonymous reviewers for their helpful feedback, and J\'er\'emie Chalopin, Mika G\"o\"os, and Joel Kaasinen for discussions and comments.

This work was supported in part by Academy of Finland (grants 129761, 132380, 132812, and 252018), the research funds of University of Helsinki, and Finnish Cultural Foundation. Part of this work was conducted while Tuomo Lempi\"ainen was with the Department of Information and Computer Science at Aalto University.

\def\UrlFont{\sf\footnotesize}
\bibliographystyle{plainnat}
\bibliography{paper}

\begin{thebibliography}{62}
\providecommand{\natexlab}[1]{#1}
\providecommand{\url}[1]{\texttt{#1}}
\expandafter\ifx\csname urlstyle\endcsname\relax
  \providecommand{\doi}[1]{doi: #1}\else
  \providecommand{\doi}{doi: \begingroup \urlstyle{rm}\Url}\fi

\bibitem[Afek et~al.(2011)Afek, Alon, Bar-Joseph, Cornejo, Haeupler, and
  Kuhn]{afek11beeping}
Yehuda Afek, Noga Alon, Ziv Bar-Joseph, Alejandro Cornejo, Bernhard Haeupler,
  and Fabian Kuhn.
\newblock Beeping a maximal independent set.
\newblock In \emph{Proc.\ 25th International Symposium on Distributed Computing
  (DISC 2011)}, volume 6950 of \emph{Lecture Notes in Computer Science}, pages
  32--50. Springer, 2011.
\newblock
  \href{http://dx.doi.org/10.1007/978-3-642-24100-0_3}{\nolinkurl{doi:10.1007/978-3-642-24100-0_3}}.

\bibitem[Angluin(1980)]{angluin80local}
Dana Angluin.
\newblock Local and global properties in networks of processors.
\newblock In \emph{Proc.\ 12th Annual ACM Symposium on Theory of Computing
  (STOC 1980)}, pages 82--93. ACM Press, 1980.
\newblock
  \href{http://dx.doi.org/10.1145/800141.804655}{\nolinkurl{doi:10.1145/800141.804655}}.

\bibitem[{\AA}strand and Suomela(2010)]{astrand10vc-sc}
Matti {\AA}strand and Jukka Suomela.
\newblock Fast distributed approximation algorithms for vertex cover and set
  cover in anonymous networks.
\newblock In \emph{Proc.\ 22nd Annual ACM Symposium on Parallelism in
  Algorithms and Architectures (SPAA 2010)}, pages 294--302. ACM Press, 2010.
\newblock
  \href{http://dx.doi.org/10.1145/1810479.1810533}{\nolinkurl{doi:10.1145/1810479.1810533}}.

\bibitem[{\AA}strand et~al.(2009){\AA}strand, Flor{\'e}en, Polishchuk, Rybicki,
  Suomela, and Uitto]{astrand09vc2apx}
Matti {\AA}strand, Patrik Flor{\'e}en, Valentin Polishchuk, Joel Rybicki, Jukka
  Suomela, and Jara Uitto.
\newblock A local 2-approximation algorithm for the vertex cover problem.
\newblock In \emph{Proc.\ 23rd International Symposium on Distributed Computing
  (DISC 2009)}, volume 5805 of \emph{Lecture Notes in Computer Science}, pages
  191--205. Springer, 2009.
\newblock
  \href{http://dx.doi.org/10.1007/978-3-642-04355-0_21}{\nolinkurl{doi:10.1007/978-3-642-04355-0_21}}.

\bibitem[{\AA}strand et~al.(2010){\AA}strand, Polishchuk, Rybicki, Suomela, and
  Uitto]{astrand10weakly-coloured}
Matti {\AA}strand, Valentin Polishchuk, Joel Rybicki, Jukka Suomela, and Jara
  Uitto.
\newblock Local algorithms in (weakly) coloured graphs, 2010.
\newblock \href{http://arxiv.org/abs/1002.0125}{\nolinkurl{arXiv:1002.0125}}.

\bibitem[Attiya et~al.(1988)Attiya, Snir, and Warmuth]{attiya88computing}
Hagit Attiya, Marc Snir, and Manfred~K. Warmuth.
\newblock Computing on an anonymous ring.
\newblock \emph{Journal of the ACM}, 35\penalty0 (4):\penalty0 845--875, 1988.
\newblock
  \href{http://dx.doi.org/10.1145/48014.48247}{\nolinkurl{doi:10.1145/48014.48247}}.

\bibitem[Benthem(1977)]{vBenthem76}
Johan~van Benthem.
\newblock \emph{Modal Correspondence Theory}.
\newblock PhD thesis, Instituut voor Logica en Grondslagenonderzoek van de
  Exacte Wetenschappen, Universiteit van Amsterdam, 1977.

\bibitem[Blackburn et~al.(2001)Blackburn, Rijke, and Venema]{blackburn01modal}
Patrick Blackburn, Maarten~de Rijke, and Yde Venema.
\newblock \emph{Modal Logic}, volume~53 of \emph{Cambridge Tracts in
  Theoretical Computer Science}.
\newblock Cambridge University Press, Cambridge, UK, 2001.

\bibitem[Blackburn et~al.(2007)Blackburn, Benthem, and
  Wolter]{blackburn07handbook}
Patrick Blackburn, Johan~van Benthem, and Frank Wolter, editors.
\newblock \emph{Handbook of Modal Logic}, volume~3 of \emph{Studies in Logic
  and Practical Reasoning}.
\newblock Elsevier, Amsterdam, 2007.

\bibitem[Boldi and Vigna(1997)]{boldi97computing}
Paolo Boldi and Sebastiano Vigna.
\newblock Computing vector functions on anonymous networks.
\newblock In \emph{Proc.\ 4th Colloquium on Structural Information and
  Communication Complexity (SIROCCO 1997)}, pages 201--214. Carleton
  Scientific, 1997.

\bibitem[Boldi and Vigna(1999)]{boldi99computing}
Paolo Boldi and Sebastiano Vigna.
\newblock Computing anonymously with arbitrary knowledge.
\newblock In \emph{Proc.\ 18th Annual ACM Symposium on Principles of
  Distributed Computing (PODC 1999)}, pages 181--188. ACM Press, 1999.
\newblock
  \href{http://dx.doi.org/10.1145/301308.301355}{\nolinkurl{doi:10.1145/301308.301355}}.

\bibitem[Boldi and Vigna(2001)]{boldi01effective}
Paolo Boldi and Sebastiano Vigna.
\newblock An effective characterization of computability in anonymous networks.
\newblock In \emph{Proc.\ 15th International Symposium on Distributed Computing
  (DISC 2001)}, volume 2180 of \emph{Lecture Notes in Computer Science}, pages
  33--47. Springer, 2001.
\newblock
  \href{http://dx.doi.org/10.1007/3-540-45414-4_3}{\nolinkurl{doi:10.1007/3-540-45414-4_3}}.

\bibitem[Boldi et~al.(1996)Boldi, Shammah, Vigna, Codenotti, Gemmell, and
  Simon]{boldi96symmetry}
Paolo Boldi, Shella Shammah, Sebastiano Vigna, Bruno Codenotti, Peter Gemmell,
  and Janos Simon.
\newblock Symmetry breaking in anonymous networks: characterizations.
\newblock In \emph{Proc.\ 4th Israel Symposium on the Theory of Computing and
  Systems (ISTCS 1996)}, pages 16--26. IEEE, 1996.

\bibitem[Bondy and Murty(1976)]{bondy76graph-theory}
John~A. Bondy and U.~S.~R. Murty.
\newblock \emph{Graph Theory with Applications}.
\newblock North-Holland, New York, 1976.

\bibitem[Chalopin(2006)]{chalopin06phd}
J{\'e}r{\'e}mie Chalopin.
\newblock \emph{Algorithmique Distribu{\'e}e, Calculs Locaux et Homomorphismes
  de Graphes}.
\newblock PhD thesis, LaBRI, Universit{\'e} Bordeaux 1, 2006.

\bibitem[Chalopin et~al.(2006)Chalopin, Das, and Santoro]{chalopin06groupings}
J{\'e}r{\'e}mie Chalopin, Shantanu Das, and Nicola Santoro.
\newblock Groupings and pairings in anonymous networks.
\newblock In \emph{Proc.\ 20th International Symposium on Distributed Computing
  (DISC 2006)}, volume 4167 of \emph{Lecture Notes in Computer Science}, pages
  105--119. Springer, 2006.
\newblock
  \href{http://dx.doi.org/10.1007/11864219_8}{\nolinkurl{doi:10.1007/11864219_8}}.

\bibitem[Conradie(2002)]{Conradie02}
Willem Conradie.
\newblock Definability and changing perspectives: The beth property for three
  extensions of modal logic.
\newblock Master's thesis, Institute for Logic, Language and Computation,
  University of Amsterdam, 2002.

\bibitem[Cornejo and Kuhn(2010)]{cornejo10deploying}
Alejandro Cornejo and Fabian Kuhn.
\newblock Deploying wireless networks with beeps.
\newblock In \emph{Proc.\ 24th International Symposium on Distributed Computing
  (DISC 2010)}, volume 6343 of \emph{Lecture Notes in Computer Science}, pages
  148--162. Springer, 2010.
\newblock
  \href{http://dx.doi.org/10.1007/978-3-642-15763-9_15}{\nolinkurl{doi:10.1007/978-3-642-15763-9_15}}.

\bibitem[Czygrinow et~al.(2008)Czygrinow, Ha{\'n}{\'c}kowiak, and
  Wawrzyniak]{czygrinow08fast}
Andrzej Czygrinow, Micha{\l} Ha{\'n}{\'c}kowiak, and Wojciech Wawrzyniak.
\newblock Fast distributed approximations in planar graphs.
\newblock In \emph{Proc.\ 22nd International Symposium on Distributed Computing
  (DISC 2008)}, volume 5218 of \emph{Lecture Notes in Computer Science}, pages
  78--92. Springer, 2008.
\newblock
  \href{http://dx.doi.org/10.1007/978-3-540-87779-0_6}{\nolinkurl{doi:10.1007/978-3-540-87779-0_6}}.

\bibitem[Czygrinow et~al.(2011)Czygrinow, Ha{\'n}{\'c}kowiak, Krzywdzi{\'n}ski,
  Szyma{\'n}ska, and Wawrzyniak]{czygrinow11semimatching}
Andrzej Czygrinow, Micha{\l} Ha{\'n}{\'c}kowiak, Krzysztof Krzywdzi{\'n}ski,
  Edyta Szyma{\'n}ska, and Wojciech Wawrzyniak.
\newblock Brief announcement: distributed approximations for the semi-matching
  problem.
\newblock In \emph{Proc.\ 25th International Symposium on Distributed Computing
  (DISC 2011)}, volume 6950 of \emph{Lecture Notes in Computer Science}, pages
  200--201. Springer, 2011.
\newblock
  \href{http://dx.doi.org/10.1007/978-3-642-24100-0_18}{\nolinkurl{doi:10.1007/978-3-642-24100-0_18}}.

\bibitem[de~Rijke(2000)]{dRijke00}
Maarten de~Rijke.
\newblock A note on graded modal logic.
\newblock \emph{Studia Logica}, 64\penalty0 (2):\penalty0 271--283, 2000.
\newblock
  \href{http://dx.doi.org/10.1023/A:1005245900406}{\nolinkurl{doi:10.1023/A:1005245900406}}.

\bibitem[Diestel(2010)]{diestel10graph}
Reinhard Diestel.
\newblock \emph{Graph Theory}.
\newblock Springer, Berlin, 4th edition, 2010.
\newblock \url{http://diestel-graph-theory.com/}.

\bibitem[Diks et~al.(1995)Diks, Kranakis, Malinowski, and
  Pelc]{diks95anonymous}
Krzysztof Diks, Evangelos Kranakis, Adam Malinowski, and Andrzej Pelc.
\newblock Anonymous wireless rings.
\newblock \emph{Theoretical Computer Science}, 145\penalty0 (1--2):\penalty0
  95--109, 1995.
\newblock
  \href{http://dx.doi.org/10.1016/0304-3975(94)00178-L}{\nolinkurl{doi:10.1016/0304-3975(94)00178-L}}.

\bibitem[Emek et~al.(2012)Emek, Smula, and Wattenhofer]{emek12stone}
Yuval Emek, Jasmin Smula, and Roger Wattenhofer.
\newblock Stone age distributed computing, 2012.
\newblock \href{http://arxiv.org/abs/1202.1186}{\nolinkurl{arXiv:1202.1186}}.

\bibitem[Fine(1972)]{fine72in}
Kit Fine.
\newblock In so many possible worlds.
\newblock \emph{Notre Dame Journal of Formal Logic}, 13\penalty0 (4):\penalty0
  516--520, 1972.
\newblock
  \href{http://dx.doi.org/10.1305/ndjfl/1093890715}{\nolinkurl{doi:10.1305/ndjfl/1093890715}}.

\bibitem[Flocchini et~al.(2003)Flocchini, Roncato, and
  Santoro]{flocchini03computing}
Paola Flocchini, Alessandro Roncato, and Nicola Santoro.
\newblock Computing on anonymous networks with sense of direction.
\newblock \emph{Theoretical Computer Science}, 301\penalty0 (1--3):\penalty0
  355--379, 2003.
\newblock
  \href{http://dx.doi.org/10.1016/S0304-3975(02)00592-3}{\nolinkurl{doi:10.1016/S0304-3975(02)00592-3}}.

\bibitem[Flor{\'e}en et~al.(2008{\natexlab{a}})Flor{\'e}en, Hassinen, Kaski,
  and Suomela]{floreen08local}
Patrik Flor{\'e}en, Marja Hassinen, Petteri Kaski, and Jukka Suomela.
\newblock Local approximation algorithms for a class of 0/1 max-min linear
  programs, 2008{\natexlab{a}}.
\newblock \href{http://arxiv.org/abs/0806.0282}{\nolinkurl{arXiv:0806.0282}}.

\bibitem[Flor{\'e}en et~al.(2008{\natexlab{b}})Flor{\'e}en, Hassinen, Kaski,
  and Suomela]{floreen08tight}
Patrik Flor{\'e}en, Marja Hassinen, Petteri Kaski, and Jukka Suomela.
\newblock Tight local approximation results for max-min linear programs.
\newblock In \emph{Proc.\ 4th International Workshop on Algorithmic Aspects of
  Wireless Sensor Networks (Algosensors 2008)}, volume 5389 of \emph{Lecture
  Notes in Computer Science}, pages 2--17. Springer, 2008{\natexlab{b}}.
\newblock
  \href{http://dx.doi.org/10.1007/978-3-540-92862-1_2}{\nolinkurl{doi:10.1007/978-3-540-92862-1_2}}.
  \href{http://arxiv.org/abs/0804.4815}{\nolinkurl{arXiv:0804.4815}}.

\bibitem[Flor{\'e}en et~al.(2008{\natexlab{c}})Flor{\'e}en, Kaski, Musto, and
  Suomela]{floreen08approximating}
Patrik Flor{\'e}en, Petteri Kaski, Topi Musto, and Jukka Suomela.
\newblock Approximating max-min linear programs with local algorithms.
\newblock In \emph{Proc.\ 22nd IEEE International Parallel and Distributed
  Processing Symposium (IPDPS 2008)}. IEEE, 2008{\natexlab{c}}.
\newblock
  \href{http://dx.doi.org/10.1109/IPDPS.2008.4536235}{\nolinkurl{doi:10.1109/IPDPS.2008.4536235}}.
  \href{http://arxiv.org/abs/0710.1499}{\nolinkurl{arXiv:0710.1499}}.

\bibitem[Flor{\'e}en et~al.(2009)Flor{\'e}en, Kaasinen, Kaski, and
  Suomela]{floreen09max-min-lp}
Patrik Flor{\'e}en, Joel Kaasinen, Petteri Kaski, and Jukka Suomela.
\newblock An optimal local approximation algorithm for max-min linear programs.
\newblock In \emph{Proc.\ 21st Annual ACM Symposium on Parallelism in
  Algorithms and Architectures (SPAA 2009)}, pages 260--269. ACM Press, 2009.
\newblock
  \href{http://dx.doi.org/10.1145/1583991.1584058}{\nolinkurl{doi:10.1145/1583991.1584058}}.
  \href{http://arxiv.org/abs/0809.1489}{\nolinkurl{arXiv:0809.1489}}.

\bibitem[Flor{\'e}en et~al.(2010)Flor{\'e}en, Kaski, Polishchuk, and
  Suomela]{floreen10almost-stable}
Patrik Flor{\'e}en, Petteri Kaski, Valentin Polishchuk, and Jukka Suomela.
\newblock Almost stable matchings by truncating the {G}ale--{S}hapley
  algorithm.
\newblock \emph{Algorithmica}, 58\penalty0 (1):\penalty0 102--118, 2010.
\newblock
  \href{http://dx.doi.org/10.1007/s00453-009-9353-9}{\nolinkurl{doi:10.1007/s00453-009-9353-9}}.
  \href{http://arxiv.org/abs/0812.4893}{\nolinkurl{arXiv:0812.4893}}.

\bibitem[Flor{\'e}en et~al.(2011)Flor{\'e}en, Hassinen, Kaasinen, Kaski, Musto,
  and Suomela]{floreen11max-min-lp}
Patrik Flor{\'e}en, Marja Hassinen, Joel Kaasinen, Petteri Kaski, Topi Musto,
  and Jukka Suomela.
\newblock Local approximability of max-min and min-max linear programs.
\newblock \emph{Theory of Computing Systems}, 49\penalty0 (4):\penalty0
  672--697, 2011.
\newblock
  \href{http://dx.doi.org/10.1007/s00224-010-9303-6}{\nolinkurl{doi:10.1007/s00224-010-9303-6}}.

\bibitem[G{\"o}{\"o}s et~al.(2012)G{\"o}{\"o}s, Hirvonen, and
  Suomela]{goos12local-approximation}
Mika G{\"o}{\"o}s, Juho Hirvonen, and Jukka Suomela.
\newblock Lower bounds for local approximation.
\newblock In \emph{Proc.\ 31st Annual ACM Symposium on Principles of
  Distributed Computing (PODC 2012)}, pages 175--184. ACM Press, 2012.
\newblock
  \href{http://dx.doi.org/10.1145/2332432.2332465}{\nolinkurl{doi:10.1145/2332432.2332465}}.
  \href{http://arxiv.org/abs/1201.6675}{\nolinkurl{arXiv:1201.6675}}.

\bibitem[Halpern and Moses(1990)]{halpern90knowledge}
Joseph~Y. Halpern and Yoram Moses.
\newblock Knowledge and common knowledge in a distributed environment.
\newblock \emph{Journal of the ACM}, 37\penalty0 (3):\penalty0 549--587, 1990.
\newblock
  \href{http://dx.doi.org/10.1145/79147.79161}{\nolinkurl{doi:10.1145/79147.79161}}.

\bibitem[Hella et~al.(2012)Hella, J{\"a}rvisalo, Kuusisto, Laurinharju,
  Lempi{\"a}inen, Luosto, Suomela, and Virtema]{hella12weak-models}
Lauri Hella, Matti J{\"a}rvisalo, Antti Kuusisto, Juhana Laurinharju, Tuomo
  Lempi{\"a}inen, Kerkko Luosto, Jukka Suomela, and Jonni Virtema.
\newblock Weak models of distributed computing, with connections to modal
  logic.
\newblock In \emph{Proc.\ 31st Annual ACM Symposium on Principles of
  Distributed Computing (PODC 2012)}, pages 185--194. ACM Press, 2012.
\newblock
  \href{http://dx.doi.org/10.1145/2332432.2332466}{\nolinkurl{doi:10.1145/2332432.2332466}}.
  \href{http://arxiv.org/abs/1205.2051}{\nolinkurl{arXiv:1205.2051}}.

\bibitem[Immerman(1999)]{immerman99descriptive}
Neil Immerman.
\newblock \emph{Descriptive Complexity}.
\newblock Graduate Texts in Computer Science. Springer, Berlin, 1999.

\bibitem[Kuhn(2005)]{kuhn05price}
Fabian Kuhn.
\newblock \emph{The Price of Locality: Exploring the Complexity of Distributed
  Coordination Primitives}.
\newblock PhD thesis, ETH Zurich, 2005.

\bibitem[Kuhn and Wattenhofer(2006)]{kuhn06complexity}
Fabian Kuhn and Roger Wattenhofer.
\newblock On the complexity of distributed graph coloring.
\newblock In \emph{Proc.\ 25th Annual ACM Symposium on Principles of
  Distributed Computing (PODC 2006)}, pages 7--15. ACM Press, 2006.
\newblock
  \href{http://dx.doi.org/10.1145/1146381.1146387}{\nolinkurl{doi:10.1145/1146381.1146387}}.

\bibitem[Kuhn et~al.(2006)Kuhn, Moscibroda, and Wattenhofer]{kuhn06price}
Fabian Kuhn, Thomas Moscibroda, and Roger Wattenhofer.
\newblock The price of being near-sighted.
\newblock In \emph{Proc.\ 17th Annual ACM-SIAM Symposium on Discrete Algorithms
  (SODA 2006)}, pages 980--989. ACM Press, 2006.
\newblock
  \href{http://dx.doi.org/10.1145/1109557.1109666}{\nolinkurl{doi:10.1145/1109557.1109666}}.

\bibitem[Lenzen(2011)]{lenzen11phd}
Christoph Lenzen.
\newblock \emph{Synchronization and Symmetry Breaking in Distributed Systems}.
\newblock PhD thesis, ETH Zurich, January 2011.

\bibitem[Lenzen and Wattenhofer(2010)]{lenzen10mds}
Christoph Lenzen and Roger Wattenhofer.
\newblock Minimum dominating set approximation in graphs of bounded arboricity.
\newblock In \emph{Proc.\ 24th International Symposium on Distributed Computing
  (DISC 2010)}, volume 6343 of \emph{Lecture Notes in Computer Science}, pages
  510--524. Springer, 2010.
\newblock
  \href{http://dx.doi.org/10.1007/978-3-642-15763-9_48}{\nolinkurl{doi:10.1007/978-3-642-15763-9_48}}.

\bibitem[Lenzen et~al.(2010)Lenzen, Oswald, and Wattenhofer]{lenzen10what}
Christoph Lenzen, Yvonne~Anne Oswald, and Roger Wattenhofer.
\newblock What can be approximated locally?
\newblock TIK Report 331, ETH Zurich, Computer Engineering and Networks
  Laboratory, November 2010.
\newblock \url{ftp://ftp.tik.ee.ethz.ch/pub/publications/TIK-Report-331.pdf}.

\bibitem[Linial(1992)]{linial92locality}
Nathan Linial.
\newblock Locality in distributed graph algorithms.
\newblock \emph{SIAM Journal on Computing}, 21\penalty0 (1):\penalty0 193--201,
  1992.
\newblock
  \href{http://dx.doi.org/10.1137/0221015}{\nolinkurl{doi:10.1137/0221015}}.

\bibitem[Mayer et~al.(1995)Mayer, Naor, and Stockmeyer]{mayer95local}
Alain Mayer, Moni Naor, and Larry Stockmeyer.
\newblock Local computations on static and dynamic graphs.
\newblock In \emph{Proc.\ 3rd Israel Symposium on the Theory of Computing and
  Systems (ISTCS 1995)}, pages 268--278. IEEE, 1995.
\newblock
  \href{http://dx.doi.org/10.1109/ISTCS.1995.377023}{\nolinkurl{doi:10.1109/ISTCS.1995.377023}}.

\bibitem[Moran and Warmuth(1993)]{moran93gap}
Shlomo Moran and Manfred~K. Warmuth.
\newblock Gap theorems for distributed computation.
\newblock \emph{SIAM Journal on Computing}, 22\penalty0 (2):\penalty0 379--394,
  1993.
\newblock
  \href{http://dx.doi.org/10.1137/0222028}{\nolinkurl{doi:10.1137/0222028}}.

\bibitem[Moscibroda(2006)]{moscibroda06locality}
Thomas Moscibroda.
\newblock \emph{Locality, Scheduling, and Selfishness: Algorithmic Foundations
  of Highly Decentralized Networks}.
\newblock PhD thesis, ETH Zurich, 2006.

\bibitem[Naor and Stockmeyer(1995)]{naor95what}
Moni Naor and Larry Stockmeyer.
\newblock What can be computed locally?
\newblock \emph{SIAM Journal on Computing}, 24\penalty0 (6):\penalty0
  1259--1277, 1995.
\newblock
  \href{http://dx.doi.org/10.1137/S0097539793254571}{\nolinkurl{doi:10.1137/S0097539793254571}}.

\bibitem[Norris(1995)]{norris94classifying-anonymous}
Nancy Norris.
\newblock Classifying anonymous networks: when can two networks compute the
  same set of vector-valued functions?
\newblock In \emph{Proc.\ 1st Colloquium on Structural Information and
  Communication Complexity (SIROCCO 1994)}, pages 83--98. Carleton University
  Press, 1995.

\bibitem[Norris(1996)]{norris95computing}
Nancy Norris.
\newblock Computing functions on partially wireless networks.
\newblock In \emph{Proc.\ 2nd Colloquium on Structural Information and
  Communication Complexity (SIROCCO 1995)}, pages 53--64. Carleton University
  Press, 1996.

\bibitem[Peleg(2000)]{peleg00distributed}
David Peleg.
\newblock \emph{Distributed Computing: A Locality-Sensitive Approach}.
\newblock SIAM Monographs on Discrete Mathematics and Applications. Society for
  Industrial and Applied Mathematics, Philadelphia, 2000.

\bibitem[Petersen(1891)]{petersen1891dietheorie}
Julius Petersen.
\newblock Die {T}heorie der regul{\"a}ren graphs.
\newblock \emph{Acta Mathematica}, 15\penalty0 (1):\penalty0 193--220, 1891.
\newblock
  \href{http://dx.doi.org/10.1007/BF02392606}{\nolinkurl{doi:10.1007/BF02392606}}.

\bibitem[Polishchuk and Suomela(2009)]{polishchuk09simple}
Valentin Polishchuk and Jukka Suomela.
\newblock A simple local 3-approximation algorithm for vertex cover.
\newblock \emph{Information Processing Letters}, 109\penalty0 (12):\penalty0
  642--645, 2009.
\newblock
  \href{http://dx.doi.org/10.1016/j.ipl.2009.02.017}{\nolinkurl{doi:10.1016/j.ipl.2009.02.017}}.
  \href{http://arxiv.org/abs/0810.2175}{\nolinkurl{arXiv:0810.2175}}.

\bibitem[Sangiorgi(2009)]{Sangiorgi09}
Davide Sangiorgi.
\newblock On the origins of bisimulation and coinduction.
\newblock \emph{ACM Transactions on Programming Languages and Systems},
  31\penalty0 (4):\penalty0 Article 15, 2009.
\newblock
  \href{http://dx.doi.org/10.1145/1516507.1516510}{\nolinkurl{doi:10.1145/1516507.1516510}}.

\bibitem[Suomela(2010)]{suomela10eds}
Jukka Suomela.
\newblock Distributed algorithms for edge dominating sets.
\newblock In \emph{Proc.\ 29th Annual ACM Symposium on Principles of
  Distributed Computing (PODC 2010)}, pages 365--374. ACM Press, 2010.
\newblock
  \href{http://dx.doi.org/10.1145/1835698.1835783}{\nolinkurl{doi:10.1145/1835698.1835783}}.

\bibitem[Suomela(2013)]{suomela13survey}
Jukka Suomela.
\newblock Survey of local algorithms.
\newblock \emph{ACM Computing Surveys}, 45\penalty0 (2):\penalty0 24:1--40,
  2013.
\newblock
  \href{http://dx.doi.org/10.1145/2431211.2431223}{\nolinkurl{doi:10.1145/2431211.2431223}}.
  \url{http://www.cs.helsinki.fi/local-survey/}.

\bibitem[Wiese and Kranakis(2008)]{wiese08impact}
Andreas Wiese and Evangelos Kranakis.
\newblock Impact of locality on location aware unit disk graphs.
\newblock \emph{Algorithms}, 1:\penalty0 2--29, 2008.
\newblock
  \href{http://dx.doi.org/10.3390/a1010002}{\nolinkurl{doi:10.3390/a1010002}}.

\bibitem[Wolfram(1983)]{wolfram83statistical}
Stephen Wolfram.
\newblock Statistical mechanics of cellular automata.
\newblock \emph{Reviews of Modern Physics}, 55\penalty0 (3):\penalty0 601--644,
  1983.
\newblock
  \href{http://dx.doi.org/10.1103/RevModPhys.55.601}{\nolinkurl{doi:10.1103/RevModPhys.55.601}}.

\bibitem[Yamashita and Kameda(1989)]{yamashita89electing}
Masafumi Yamashita and Tsunehiko Kameda.
\newblock Electing a leader when processor identity numbers are not distinct
  (extended abstract).
\newblock In \emph{Proc.\ 3rd International Workshop on Distributed Algorithms
  (WDAG 1989)}, volume 392 of \emph{Lecture Notes in Computer Science}, pages
  303--314. Springer, 1989.
\newblock
  \href{http://dx.doi.org/10.1007/3-540-51687-5_52}{\nolinkurl{doi:10.1007/3-540-51687-5_52}}.

\bibitem[Yamashita and Kameda(1996{\natexlab{a}})]{yamashita96computing}
Masafumi Yamashita and Tsunehiko Kameda.
\newblock Computing on anonymous networks: part {I}---characterizing the
  solvable cases.
\newblock \emph{IEEE Transactions on Parallel and Distributed Systems},
  7\penalty0 (1):\penalty0 69--89, 1996{\natexlab{a}}.
\newblock
  \href{http://dx.doi.org/10.1109/71.481599}{\nolinkurl{doi:10.1109/71.481599}}.

\bibitem[Yamashita and
  Kameda(1996{\natexlab{b}})]{yamashita96computing-functions}
Masafumi Yamashita and Tsunehiko Kameda.
\newblock Computing functions on asynchronous anonymous networks.
\newblock \emph{Mathematical Systems Theory}, 29\penalty0 (4):\penalty0
  331--356, 1996{\natexlab{b}}.
\newblock
  \href{http://dx.doi.org/10.1007/BF01192691}{\nolinkurl{doi:10.1007/BF01192691}}.

\bibitem[Yamashita and Kameda(1996{\natexlab{c}})]{yamashita96computinga}
Masafumi Yamashita and Tsunehiko Kameda.
\newblock Computing on anonymous networks: part {II}---decision and membership
  problems.
\newblock \emph{IEEE Transactions on Parallel and Distributed Systems},
  7\penalty0 (1):\penalty0 90--96, 1996{\natexlab{c}}.
\newblock
  \href{http://dx.doi.org/10.1109/71.481600}{\nolinkurl{doi:10.1109/71.481600}}.

\bibitem[Yamashita and Kameda(1999)]{yamashita99leader}
Masafumi Yamashita and Tsunehiko Kameda.
\newblock Leader election problem on networks in which processor identity
  numbers are not distinct.
\newblock \emph{IEEE Transactions on Parallel and Distributed Systems},
  10\penalty0 (9):\penalty0 878--887, 1999.
\newblock
  \href{http://dx.doi.org/10.1109/71.798313}{\nolinkurl{doi:10.1109/71.798313}}.

\end{thebibliography}

\end{document}